\documentclass[]{article}
\usepackage[utf8]{inputenc}
\usepackage[english]{babel}

\usepackage{amsmath}
\usepackage{amsfonts}
\usepackage{amssymb}
\usepackage{comment}

\DeclareMathOperator{\csp}{CSP}
\DeclareMathOperator{\Aut}{Aut}
\DeclareMathOperator{\maxbound}{\mathbb{L}}
\DeclareMathOperator{\graphinstance}{\mathcal{G}_{\mathcal{I}}}
\DeclareMathOperator{\proj}{\textrm{proj}}
\DeclareMathOperator{\instance}{\mathcal{I}}

\DeclareMathOperator{\structA}{\mathbb{A}}
\DeclareMathOperator{\structB}{\mathbb{B}}
\DeclareMathOperator{\solution}{\mathbf{s}}
\DeclareMathOperator{\constraint}{\mathbf{C}}

\DeclareMathOperator{\V}{\mathcal{V}}
\DeclareMathOperator{\indices}{\mathbb{I}}
\DeclareMathOperator{\bipartite}{\mathcal{B}}
\DeclareMathOperator{\assignment}{\texttt{a}}

\DeclareMathOperator{\vertices}{\mathbf{Vert}}

\DeclareMathOperator{\Q}{\mathbb{Q}}
\DeclareMathOperator{\N}{\mathbb{N}}
\DeclareMathOperator{\bounds}{\mathcal{F}}

\newenvironment{proof}{\paragraph{Proof:}}{\hfill$\square$}
\begin{document}
\newcounter{thm}
\newtheorem{theorem}[thm]{Theorem}
\newtheorem{lemma}[thm]{Lemma}
\newtheorem{observation}[thm]{Observation}
\newtheorem{proposition}[thm]{Proposition}
\newtheorem{definition}[thm]{Definition}
\newtheorem{corollary}[thm]{Corollary}
\title{On The Relational Width of  First-Order Expansions of  Finitely Bounded Homogeneous Binary Cores with Bounded Strict Width\footnote{The research was partially supported by NCN grant number 2014/14/A/ST6/00138.}}
\author{Micha{\l} Wrona\\
Theoretical Computer Science Department\\
 Jagiellonian University\\
Poland\\
\texttt{michal.wrona@uj.edu.pl}          %% \email is recommended
}


\maketitle

\begin{abstract}
 The relational width of a finite structure, if bounded, is always $(1,1)$ or $(2,3)$. In this
paper we study the relational width of first-order expansions of
finitely bounded homogeneous binary cores where binary cores are structures with equality and some anti-reflexive binary relations such that for any two different elements $a, b$ in the domain there is 
exactly one binary relation $R$ with  $(a, b) \in R$. 

Our main result is that first-order expansions of liberal finitely bound\-ed homogeneous binary cores 
with bounded strict width have relational width
(2, MaxBound) where MaxBound
is the size of the largest forbidden substructure, but is not less than $3$,
and liberal stands for  structures that do not forbid certain finite structures of small size. This result is built on a new approach and concerns a broad class of structures including reducts of homogeneous digraphs for which the CSP complexity classification has not yet been obtained.
\end{abstract}




\section{Introduction}

An instance of the \emph{constraint satisfaction problem (CSP)} consists of a number of variables and a number of  local restrictions on variables called \emph{constraints}. The question is whether there exists a global assignment to variables that satisfies all constraints. The CSP naturally generalizes SAT, expresses a number of other natural problems including $k$-coloring, solving equations over finite fields but, at least among theoreticians, is  associated mainly to the question on dichotomy~\cite{FederVardi}, i.e., is every right-hand side restriction  $\csp(\structB)$
of the CSP in P or NP-complete?  A relational structure $\structB$, known also as a (constraint) language or a template, restricts the available constraints to these that can be modelled by relations in $\structB$. It is known already for a while that the dichotomy for $\csp(\structB)$ over finite structures exists. Indeed, the Feder Vardi conjecture, on the existence of the dichotomy, has been confirmed by Zhuk~\cite{ZhukDichotomy} and independently by Bulatov~\cite{BulatovDichotomy}.

The problem is that already a very simple scheduling problem with precedence constraints of the form $(X < Y)$ cannot be properly expressed as $\csp(\structB)$ if the domain of $\structB$ is finite and scheduling problems are what practitioners think of when they hear of the CSP.  In order to express a richer class of problems including many scheduling problems as well as problems in spatial and temporal reasoning~\cite{BodirskyJonsson-survey} one considers \emph{$\omega$-categorical} structures $\structB$ that,  although infinite, share many nice properties with finite structures. In particular they admit a simple finite representation and the algebraic approach to the complexity of CSP which stands behind the both dichotomy proofs is also applicable in this context. 
(The precise definition of $\omega$-categoricity as well as many other well-known notions used in the introduction are defined formally in the remainder of the paper, usually in Section~\ref{sect:preliminaries}.)
It is no dichotomy to look for among all $\omega$-categorical CSPs (CSPs over $\omega$-categorical templates)~\cite{nondichotomies-omegacat}.  
Thus, one considers a subclass --- (first-order) reducts of (countably infinite)  finitely bounded homogeneous structures for which a dichotomy is conjectured. In what follows we will call this conjecture the \emph{infinite dichotomy conjecture}.
In contrast to $\omega$-categorical CSPs, all CSPs over  reducts of finitely bounded homogeneous structures are in NP and more importantly, the infinite dichotomy conjecture 
is backed by an algebraic dichotomy~\cite{algebraic-dichotomy-omegacat} delineating algebras corresponding to  structures 
with no non-trivial symmetries and such algebras with non-trivial symmetries. 
It is already known that algebras with no non-trivial symmetries correspond to NP-complete problems. The \emph{infinite tractability conjecture} states that a $\csp(\structB)$ is in P (tractable) always when the corresponding algebra contains some non-trivial symmetries (polymorphisms). 
A similar tractability conjecture  concerning finite algebras was confirmed by Bulatov and Zhuk. 

While in the finite case, the algorithm solving tractable $\csp(\structB)$ is based on two prevailing general algorithmic techniques: local-consistency methods~\cite{BoundedWidth} and the 'few subpowers' algorithm~\cite{Idziak}, the
development of general algorithmic techniques and  establishing 
the limits of their applicability in the infinite case are rather in their infancy. 
Two important  exceptions are: an algebraic characterization of $\omega$-categorical structures with \emph{bounded strict width}~\cite{OligoClone,Datalog-omegacat} and the lifiting theorem~\cite{UnaryDichotomy} which lifts the tractability from finite CSP. Since it is already known that the tractability of some tractable reducts of $(\Q;<)$ a.k.a. temporal languages cannot be explained by the lifiting theorem, the development of general algoritmic techniques for the CSPs within the scope of the infinite tractability conjecture and understanding the limits of their applicability seems inevitable. 
Natural research questions in this context concern local-consistency methods.
Firstly, because the algebraic characterization of finite structures whose CSP may be solved in this way is considered to be an important step towards establishing the dichotomy. Secondly, because local consistency methods are ubiquitous
in constraint solving, for instance, in the context of qualitative calculi in spatial and temporal reasoning~\cite{Renz12}. In this paper we consider one of these natural questions. 

A structure $\structB$ has \emph{bounded width} if $\csp(\structB)$ is solvable by the local-consisten\-cy algorithm.  Equivalently, $\csp(\structB)$ is solvable by an algorithm establishing $(k,l)$-minimiality for some natural numbers $k \leq l$.
In this case we say that $\structB$ has relational width $(k,l)$ and if such $k,l$ exist that $\structB$ has bounded relational width~\cite{BRWHierarchy,BulatovRW}. The relational width of $\structB$ proved to be a natural way to measure the amount of consistency needed to solve $\csp(\structB)$. In particular, it is known that if the relational width of a finite structure is bounded, then it is $(1,1)$ or $(2,3)$~\cite{BRWHierarchy}. This characterization is based on the algebraical characterization of finite structures with bounded (relational) width, and since its counterpart for reducts of finitely bounded homogeneous structures does not exist, it could be very hard to answer the question of what is the relational width of these structures. But as we already mentioned, there is such a characterization for structures with bounded strict width.
A structure $\structB$ has strict width $k$  if every partial solution to every $(k,l)$-minimal instance of $\csp(\structB)$ can be extended to a total solution. This notion is not only of theoretical interest~\cite{FederVardi} but also
under the name \emph{local-to-global consistency} has been studied in constraint solving in spatial and temporal reasoning, see e.g.~\cite{Dechter}.

In this paper we characterize the relational width of first-order expansions of 
liberal finitely bounded homogeneous binary cores with bounded strict width.
All the definitions that are necessary to understand the result are given in the following subsection.

\subsection{Results}

We say that a structure $\structA$ over a relational signature (here assumed to be finite) is homogeneous if every local isomorphism  between finite substructures of $\structA$ may be extended to an automorphism of $\structA$.
A structure $\structA$ over a signature $\tau$  is finitely bounded if there exists a finite set of finite $\tau$-structures $\bounds_{\structA}$ such that a finite structure $\Delta$ embeds into $\structA$ if and only if there is no $\Gamma$ in   $\bounds_{\structA}$ that embeds into $\Delta$. 

All homogeneous graphs, classified in~\cite{LachlanWoodrow}, and many homogenenous digraphs, classified in~\cite{HomoCherlin} , are finitely bounded.
In particular, it is known~\cite{Henson} that for any countable set of pairwise non-embedabble tournaments $\bounds$ there exists a homogeneous digraph $\structA$ such that $\bounds = \bounds_{\structA}$. Such homogeneous digraphs are known as Henson digraphs. 
In this paper we see homogeneous graphs and homogeneous digraphs over an extended signature and study these structures and many other as binary cores defined in what follows.
 
\paragraph{Binary Cores.} We say that a structure $\structA$ over domain $A$ is a \emph{binary core} 
if its signature besides $=$ contains only binary anti-reflexive relations	 $R_1, \ldots, R_{\kappa}$ such that for any two different elements $a, b \in A$
there is exactly one $R_i$ with $i \in [\kappa]$ such that 
$(a, b) \in R_i$. 

\paragraph{Examples.} A perfect example of a finitely bounded homogeneous binary core is 
a homogeneous graph seen over
the signature $\{ E, N, = \}$ where $N$ contains all different pairs of elements which are not  connected by an edge $E$, i.e., $(N(x,y) \equiv (\neg E(x,y) \wedge x \neq y))$. Other examples are homogeneous digraphs seen over the signature $\{ \curvearrowright, N, = \}$ where $\curvearrowright$ stands for an arc and $N$ is a non-arc relation, i.e., $(N(x,y) \equiv \neg \curvearrowright(x,y) \wedge \neg \curvearrowright (y,x) \wedge x \neq y)$.

\paragraph{Liberal Structures.}
We restrict ourselves to these binary cores $\structA$ that are additionally \emph{liberal}, i.e., $\bounds_{\structA}$ contains no finite structures of size $3,4,5$, or $6$. In particular any  Henson digraph that forbids tournaments of size $7$ or greater only, or a random  graph $\mathcal{G}$ seen over the extended signature is a liberal binary core. Indeed,
all structures in $\bounds_{\mathcal{G}}$ are of size less than $3$.

\paragraph{The Main Result.} We write $\maxbound_{\structA}$ to denote the maximum of $3$ and  the size of the largest structure in $\bounds_{\structA}$. A first-order expansion $\structB$ of $\structA$ is an expansion of $\structA$ such that all relations in $\structB$ have first-order definitions in $\structA$.
We are now in the position to formulate the main result of this paper.

\begin{theorem}
\label{thm:main}
Let $\structA$ be a liberal finitely bounded homogeneous binary core and $\structB$ a first-order expansion of $\structA$ with bounded strict width. Then $\structB$ has relational width $(2, \maxbound_{\structA})$.
\end{theorem}

Examples of first-order expansions of the random graph with bounded strict width were given in~\cite{WronaSTACS2020}.

\begin{proposition}
[\cite{WronaSTACS2020}]
\label{prop:ClausesQNU}
Let the structure $\structB$ be a first-order expansion of the structure $(A; E, N,$
$ =)$  where $(A;E)$ is  the random graph  
such that every relation in $\structB$
is pp-definable as a conjunction of clauses of the form:
\begin{align}
(x_1 \neq y_1 \vee \cdots \vee x_k \neq y_1 \vee R(y_1, y_2)  \vee y_2 \neq z_1 \vee\nonumber\\
\vee \cdots \vee y_2 \neq z_l), \nonumber
\end{align}
where $R \in \{ E, N \}$.
Then $\mathbb{A}$ has bounded strict width.  
\end{proposition}

Further, by Lemma~8 in~\cite{ecsps} and Theorem~1 in~\cite{CPWrona12}, we have that an equality language, i.e., a first order expansion  of $(\N; =, \neq)$ has bounded strict width if and only if all the extra relations are pp-definable by a conjunction of disjunctions of disequalities, i.e., clauses of the form: 
\begin{align}
(x_1 \neq y_1 \vee \cdots \vee x_k \neq y_k) \nonumber.
\end{align}
%such that not all $x_i, y_i$ with $i \in [k]$ have to be pairwise different. 

In order to prove Theorem~\ref{thm:main} we show that all instances of $\csp(\structB)$ under consideration are simple in a particular sense.
Roughly speaking, for a $(2, \maxbound_{\structA})$-minimal instance $\instance$ of $\csp(\structB)$ we construct a digraph $\graphinstance$ over pairs $((v,x), C)$ where $v,x$
are variables  in the instance $\instance$ and $C \subsetneq A^2$ is a binary relation \emph{primitively-positively definable} in  $\structB$. 
(Primitive positive (pp-)definitions are provided by \emph{primitive-positive (pp-)formulas} which are first-order formulas built out of the conjunction, existential quantifiers and atomic formulae only.) 
There is an arc from 
$((v_1, x_1), C)$ to $((v_2, x_2), D)$  in $\graphinstance$ if there is a constraint over a relation $R$ whose scope $(y_1, \ldots, y_k)$ contains $v_1, x_1, v_2, x_2$ and $R(y_1, \ldots, y_k)$  entails 
$(C(v_1, x_1) \implies D(v_2, x_2))$ so that $C$ and $D$ are not the whole projections of $R(y_1, \ldots, y_k)$ to $(v_1, x_1)$ and  $(v_2, x_2)$, respectively.
We say that $\structB$ is implicationally simple if every $\graphinstance$ of every $(2, \maxbound_{\structA})$-minimal instance $\instance$ of $\csp(\structB)$ is acyclic. We show that
\begin{itemize} 
\item all implicationally simple first-order expansions of finitely bounded homogeneous binary cores $\structA$ have relational width $(2, \maxbound_{\structA})$ (in Section~\ref{sect:ImplicationallySimple}) and
\item  all implicationally hard  first-order expansions  of liberal finitely bound\-ed homogeneous binary cores have no bounded strict width (in Sections~\ref{sect:CriticalTernary},\ref{sect:EffEntImp} and~\ref{sect:LiberalBinaryCores}).
\end{itemize}

We would like to mention that not all first-order expansions of finitely bounded homogeneous binary cores with bounded strict width are  implicationally simple. In~\cite{WronaSTACS2020}
one can find the following.

\begin{proposition}
[\cite{WronaSTACS2020}]
\label{prop:qnuExTwoEquivClass}
The structure $\structB = (A; E, N, =, R)$ where 
$(A; E)$ is $C^{\omega}_2$ (two disjoint infinite cliques) and $$R(x_1, x_2, x_3) \equiv ((E(x_1, x_2) \wedge N(x_2, x_3)) \vee (N(x_1, x_2) \wedge E(x_2, x_3)))$$ has bounded strict width. 
\end{proposition}

Clearly, $\structB$ is not implicationally simple. Already $\graphinstance$ for an instance $\instance$ with one constraint over relation $R$ is not acyclic.

%On the other hand, there exist finitely bounded homogeneous binary cores without bounded strict width that are implicationally simple, e.g., 
%$(\Q, \{  (x_1, x_2, x_3) \mid x_1 =   \})$

\subsection{Related Work}

In~\cite{WronaSTACS2020}, the following was proved.

\begin{theorem}
\label{thm:mainSTACS}
(\cite{WronaSTACS2020})
Let $\structA = (A; E, N, =)$ be such that $(A;E)$ is a homogeneous graph and $\structB$ a first-order expansion of $\structA$ with bounded strict width. Then
$\structB$ has relational width  $(2, \maxbound_{\structA})$. 
\end{theorem}

Clearly, Theorem~\ref{thm:main} generalizes Theorem~\ref{thm:mainSTACS} when $\structA$ is liberal but more importantly, the proof in~\cite{WronaSTACS2020} is based on the 
complexity classification of CSPs over first-order reducts of homogeneous graphs in~\cite{homographs-siam}, while the proof in this paper 
is based only on the assumptions that 
\begin{itemize}
\item
$\structA$ is a liberal finitely bounded homogeneous binary core and that
\item $\structB$ has bounded strict width. 
\end{itemize}
In particular our result concerns all but finitely many finitely bounded Henson digraphs and many other structures for which CSP classifications has not been provided.

\section{Preliminaries}
\label{sect:preliminaries}
We write $[n]$ for $\{ 1, \ldots, n\}$ and when $t$ is an $n$-tuple we write $t[i]$ with $i \in [n]$ to denote the $i$-th value in $t$. 

\subsection{Structures under consideration}
We consider here (countably infinite) finitely bounded homogeneous relational structures over domain $A$ usually denoted by $\structA$ and their first-order expansions $\structB$ also over domain $A$. In particular we look at first-order expansions $\structB$ of (liberal) finitely bounded homogeneous  
binary cores $\structA$. Recall that liberal means that $\bounds_{\structA}$ contains no structures of size $3,4,5$, and $6$.
For the sake of simplicity we write $R$ to denote both a relational symbol in a signature of $\structB$ as well as the actual relation $R^{\structB}$.

We will write $\proj_{i_1, \ldots, i_k} R$ for an $n$-ary relation $R$ and 
$\{ i_1, \ldots, i_k \} \subseteq [n]$
 to denote 
the relation $R'(y_{1}, \ldots, y_k)$ defined by the formula 
$$\exists x_1 \cdots \exists x_n~R(x_1, \ldots, x_n) \wedge \bigwedge_{j \in [k]} y_j = x_{i_j}.$$
 
We write $\Aut(\structA)$ to denote the set of automorphisms of $\structA$.
An \emph{orbit of a tuple} $t$ with values in $A$ wrt. $\Aut(\structA)$ is the set 
$\{ (\alpha(t[1]), \ldots, \alpha(t[n])) \mid \alpha \in \Aut(A) \}$. 
When $\structA$ is known from the context we simply say an orbit of a tuple instead of an orbit of a tuple wrt. $\Aut(\structA)$. We also say that $O$ is an orbit 
if it is an orbit of some tuple. An \emph{orbital} is an orbit of a tuple with two values.

All the structures under consideration are $\omega$-categorical, i.e., their first-order theories have one countable model up to isomorphism.
By the theorem proved independently by Ryll-Nardzewski, Engeler and Svenonius, 
a structure $\structB$ is $\omega$-categorical if and only if its automorphism group is oligomorphic, i.e., for every $n$ the number of orbits of $n$-tuples is finite. 
Since $\structA$ is homogeneous, we have the following.

\begin{observation}
\label{obs:BinaryCoreOrbitals}
Let $\structA = (A; R_1, \ldots, R_\kappa, =)$ be a binary core. Then  $R_i$ for all $i \in [\kappa]$ is an orbital.
\end{observation}

\begin{proof}
The equality is clearly an orbital. 
Since $\structA$ is homogeneous we have  that there is an automorphism $\alpha \in \Aut(\structA)$ such that $(\alpha(a_1), \alpha(a_2)) = (a_3, a_4)$ whenever $(a_1, a_2), (a_3, a_4) \in R_i$ and $i \in [\kappa]$. It follows that every $R_i$ is a subset of some orbital wrt. $\Aut(\structA)$. On the other hand, automorphism do not send $(a_1, a_2) \in R_i$ to $(a_3, a_4) \in R_j$ for $i \neq j$. The observation follows.
\end{proof}

Further, since all structures under consideration have quantifier elimination, i.e., 
all first-order definable relations are definable without quantifiers, we have that every binary relation fo-definable in a binary core $\structA$ is a union of  
orbitals $R_i$ with $i \in [\kappa]$. For a binary relation $C \subseteq A^2$ we will write $C^{-1}$ to denote $(C^{-1}(x,y) \equiv C(y,x))$.
We say that a binary relation  $C$ is  anti-reflexive if it is contained in $\neq$ or in other words if for all $(a,b) \in C$ we have that $a \neq b$.

It happens that all liberal finitely bounded homogeneous binary cores pp-define $\neq$.

\begin{observation}
\label{obs:NeqCanBeDefined}
Let $\structA$ be a liberal finitely bounded homogeneous 
binary core.
Then $\structA$ pp-defines $\neq$.
\end{observation}

\begin{proof}
If $\neq$ is an orbital wrt $\Aut(\structA)$, then we are done. Otherwise,  there are at least two different orbitals $O_1,O_2 \subseteq \neq$. 
We claim that 
$(x_1 \neq x_2)$  is pp-defined by:  $$(\psi(x_1, x_2) \equiv (\exists x_0~O_1(x_0, x_1) \wedge O_2(x_0, x_2))).$$ Indeed, since $O_1$ and $O_2$ are different we have that  an assignment $\assignment: \{ x_1, x_2 \} \to A$ such that $\assignment(x_1) = \assignment(x_2)$ does not satisfy $\psi$. On the other hand, since $\structA$ is liberal, for all anti-reflexive orbitals $O_3$ there exist element $a_0, a_1, a_2 \in A$ such that $(a_0, a_1) \in O_1, (a_0, a_2) \in O_2$ and $(a_1, a_2) \in O_3$. It implies that an assignment $\assignment: \{ x_1, x_2 \} \to A$ such that $\assignment(x_i) = a_i$ for $i \in [2]$ satisfies $\psi$.  Since $O_3$ was chosen arbitrarily, we have that $\psi$ is a definition of $\neq$.
%If $\neq$ is an orbital wrt $\Aut(\structA)$, then we are done. Otherwise,  there are at least two different orbitals $O,P \subseteq \neq$. Since $\structA$ is liberal, it is easy to see that
%$$(x_1 \neq x_2) \equiv (\exists y~O(y, x_1) \wedge P(y, x_2)).$$  
\end{proof}

In the paper, ternary and quaternary ($4$-ary) relations are of special interest. 
We will say that a tuple $t = (t[1], t[2], t[3])$ and $t = (t[1], t[2], t[3], t[4])$ are
$OP$-tuples for some orbitals $O,P$  if $(t[1], t[2]) \in O$, $(t[2], t[3]) \in P$ and 
$(t[1], t[2]) \in O$, $(t[3], t[4]) \in P$, respectively.

Further, we will say that a tuple is \emph{constant} if all its values are the same and that is \emph{non-constant} otherwise. A tuple is \emph{injective} if all its values are pairwise different.

\subsection{Entailment}
A first-order formula $\varphi(x_1, \ldots, x_n)$
\emph{entails} a first-order formula $\psi(x_1, \ldots, x_n)$ if the formula
$$(\forall x_1 \cdots \forall x_n~(\varphi(x_1, \ldots, x_n) \implies \psi(x_1, \ldots, x_n)))$$ is valid.
Further, we will say that an $n$-ary relation $R$ entails a formula $\psi$ over variables $\{ x_1, \ldots, x_n \}$ if $R(x_1, \ldots, x_n)$ entails $\psi(x_1, \ldots, x_n)$.

We also say that an $n$-ary relation $R$ \emph{entails no equalities} if there are no different $i,j \in [n]$ such that $R$ entails $(x_i = x_j)$.

\subsection{CSP}
A \emph{constraint} $\constraint$ is a pair $((x_1, \ldots, x_k), R)$ where $(x_1, \ldots, x_k)$ is the $k$-tuple of variables called also the \emph{scope} of the constraint  and $R$ is a $k$-ary relation. We will write $\proj_{x_{i_1}, \ldots, x_{i_l}}  \constraint $ for the \emph{projection} of a costraint $\constraint := ((x_1, \ldots, x_k), R)$ to  a tuple of variables $(x_{i_1}, \ldots, x_{i_l})$ with 
$\{  i_1, \ldots, i_l \} \subseteq [k]$. We will have that  $\proj_{x_{i_1}, \ldots, x_{i_l}}  \constraint $ is the constraint $((x_{i_1}, \ldots, x_{i_l}), R')$ where $R' = \proj_{i_1, \ldots, i_l} R$.

We study the problem $\csp(\structB)$ parametrized by first-order expansions of finitely bounded homogeneous structures. The instance $\instance$ 
of $\csp(\structB)$ is a set of constraints $((x_1, \ldots, x_k), R)$ such that $R$ is a relation in $\structB$. We say that $\instance$ is over variables $\V$ if for every constraint $((x_1, \ldots, x_k), R)$ in $\instance$ we have that $x_1, \ldots, x_k \in \V$. The question in the problem $\csp(\structB)$ is whether there exists a solution to $\instance$, i.e., an assignment $\solution: \V \to A$ to variables in $\instance$ such that for all constraints $((x_1, \ldots, x_k), R)$ we have 
$(\solution(x_1), \ldots, \solution(x_k)) \in R$.

Let $W \subseteq \V$. An assignment $\textbf{a}: W \rightarrow A$ is a partial solution to $\mathcal{I}$ if $\textbf{a}$ satisfies all projections of constraints in $\mathcal{I}$ to variables in $W$.
We will sat that $\instance$ entails no equalities if no relations in the constraints of $\instance$ do.

%\begin{proposition}
%\label{prop:PPExp}
%Let $\structB = (A; R_1, \ldots,R_l)$ be a relational structure, and let $R$ be a relation that has a primitive-positive definition in $\structB$. Then $\csp(\mathbb{B})$ and $\csp(A, R,R_1,\ldots,R_l)$ are log-space equivalent.
%\end{proposition}

\subsection{The universal-algebraic approach}

We say that an operation $f :A^n \rightarrow A$ is a \emph{polymorphism} of an $m$-ary relation $R$ iff for all $m$-tuples $t_1,\ldots,t_n \in  R$, it holds that  the tuple $(f(t_1[1],\ldots,t_n[1]),\ldots,$
$f(t_1[m],\ldots,t_n[m]))$ is also in $R$. We  will write $f(t_1,\ldots,t_n)$ as a shorthand for the expression $(f(t_1[1],\ldots,t_n[1]),\ldots,f(t_1[m],\ldots,t_n[m]))$. An operation $f$ is a polymorphism of $\mathbb{A}$ if it is a polymorphism of every relation in $\mathbb{A}$. 
If $f : A^n \rightarrow A$ is a polymorphism of $\mathbb{A}$, $R$, we say that $f$ preserves $\mathbb{A}, R$.
%, otherwise that $f$ violates $\mathbb{A}, R$.
A set of polymorphisms of an $\omega$-categorical structure $\mathbb{A}$ forms an algebraic object called an oligomorphic locally closed clone~\cite{OligoClone},
which in particular contains an oligomorphic permutation group~\cite{oligo}.

Recall that a first-order formula is a \emph{primitive-positive formula (pp-formula)}
if it is built out of conjunction, existential quantifiers and atomic formulae  only. There is a deep connection between the polymorphisms of a structure and the relations pp-definable in that structure.

\begin{theorem}
(\cite{BodirskyN06}) 
\label{thm:Galoisconn}
Let $\mathbb{A}$ be a countable $\omega$-categorical structure. Then $R$ is preserved by the polymorphisms of $\mathbb{A}$ if and only if it has a  primitive-positive definition in $\mathbb{A}$.
%, i.e., a definition via a pp-formula.
\end{theorem}

We say that a set of operations $F$ generates a set of operations $G$ if every $g \in G$ is in the smallest locally-closed clone containing $F$.
%We wite $\overline{\textrm{Aut}(\mathbb{A})}$ to denote the clone generated by the automorphisms of the structure $\mathbb{A}$.
An operation $f$ of an oligomorphic clone $F$ is called \emph{oligopotent} if $\{ g \}$ where $g(x) := f(x, \ldots, x)$ 
is generated by the permutations in $F$.
We say that a $k$-ary operation $f$ over domain $A$ is a \emph{quasi near-unanimity operation (qnu-operation)} if 
\begin{align}
f(y, x, \ldots, x) = f(x,y,x, \ldots, x) = \cdots = \nonumber\\
\cdots = f(x, \ldots, x, y) = f(x, \ldots, x) \nonumber
\end{align}
 for all
$x,y \in A$.

\subsection{Widths and Minimality}
We will now give a formal defnition of a $(k,l)$-minimal instance.

\begin{definition}
We say that an instance $\instance$ over $\V$ of $\csp(\structB)$ is $(k,l)$-minimal with $k \leq l$ if  both of the following hold:
\begin{itemize}
\item every subset of at most $l$ variables in  $\V$ is contained in a scope of some constraint in $\instance$ and
\item for every at most $k$-element subset of variables $X = \{ x_1, \ldots, x_k \} \subseteq \V$ and any two constraints $C_1, C_2 \in \instance$ whose scopes contain $X$ the projections $\proj_{x_1, \ldots, x_k} C_1$ and
$\proj_{x_1, \ldots, x_k} C_2$ are the same.
\end{itemize}
\end{definition}

We say that an instance $\instance$ of the CSP is \emph{non-trivial} if it does not contain a constraint $((x_1, \ldots, x_k), \emptyset)$.  Otherwise, $\instance$ is \emph{trivial}.

Set $k \leq l$. Clearly not every instance $\instance$ over variables $\V$ of $\csp(\structB)$ for $\structB$ over domain $A$ is $(k,l)$-minimal, however, the algorithm that obtains an equivalent $(k,l)$-minimal instance is straightforward and works in time $O(\left| \V \right|^m)$ where $m$ is the maximum of $l$ and the largest arity in the signature of $\structB$. Indeed, it is enough to introduce a new constraint $((x_1, \ldots, x_l), A^l)$ for all pairwise different variables $x_1, \ldots, x_l \in \V$ to satisfies the first condition. Then the algorithm removes tuples (orbits) from  the relations in constraints in the instance as long as the second condition is not satisfied. It is widely known and easy to prove that an instance $\instance'$ of the CSP obtained by the described algorithm is equivalent, i.e., has the same set of solution, to the orginal instance $\instance$. In particular we have that  if $\instance'$ is trivial, then $\instance$ has no  solutions.
Under a natural assumption that $\structB$ contains all at most $l$-ary relations $pp$-definable in $\structB$, we have that $\instance'$ is an instance of $\csp(\structB)$. From now on  this assumption will be in effect.

\begin{definition}
A relational structure $\structB$ has \emph{relational width $(k,l)$} if every $(k,l)$-minimal instance $\instance$ of $\structB$ has a solution iff it is non-trivial.

A relational structure $\structB$ has \emph{bounded relational width} if it has relational width $(k,l)$ for some natural numbers $k \leq l$.
\end{definition}

In this paper  we mainly look at $(2, \maxbound_{\structA})$-minimal instances $\instance$ of $\csp(\structB)$  for first-order expansions $\structB$ of finitely bounded homogeneous binary cores $\structA$. For such instances $\instance$ and $x,y \in \V$ we will write $\instance_{x, y}$ to denote the projection of any constraint in $\instance$ to the variables $(x,y)$. 
%We will asume without loss of generality that $\instance$ entails no equalities, i.e., there are no different $x,y \in V$ such that $\instance_{x,y}$ is $=$. 

\noindent
We now turn to strict width.

\begin{definition}
We say that $\structB$ has strict width $k$ if there exists $l$ such that every partial solution of every $(k,l)$-minimal instance of $\csp(\structB)$ may be extended to a total solution. 
\end{definition}
 
 The following theorem provides a characterization of bounded strict width that we use intensively in this paper.
 
 \begin{theorem}
\cite{Datalog-omegacat,OligoClone}
\label{thm:strictwidth}
Let $\mathbb{B}$ be an $\omega$-categorical language. Then  the following are equivalent.
\begin{enumerate}
 \item $\mathbb{B}$ has strict width $k$.
 \item $\mathbb{B}$ has an oligopotent $(k+1)$-ary quasi near-unanimity operation as a polymorphism.
% \item Every primitive positive formula is in $\mathbb{A}$ equivalent to a conjunction of at most $(k-1)$-ary primitive positive formulas. 
\end{enumerate}
\end{theorem}
 
Observe that in order to show that some structure $\structB$ has no bounded strict width it is enough to show that that they are not preserved by any oligopotent qnu-operations or, by Theorem~\ref{thm:Galoisconn} to pp-define
a  structure $\structB'$ of which we already know that has no bounded strict width.

\section{Critical Ternary Relations}
\label{sect:CriticalTernary}

We define a family of structures which is the main source of 'infinite' strict width, i.e., whenever we want to show that some structure does not have bounded strict width we pp-define a \emph{critical ternary relation}.

\begin{definition}
\label{def:CriticalTernary}
%Let $\structA$ be a finitely-bounded homogeneous edge-colored clique and $\structB$ a first-order expansion of $\structA$.
We say that a relation $R$   is a \emph{critical ternary relation} 
over $(\structB, C_1, C_2, D_1, D_2)$ if all 
of the following hold:
\begin{itemize}
\item $R, C_1, C_2, D_1, D_2$ are 
pp-definable in $\structB$,
\item $C_1$ and $C_2$ are disjoint and contained in $ \proj_{1,2} R$ 
\item $D_1$ and $D_2$ are disjoint and contained in $ \proj_{2,3} R$,
\item  both $C_1, D_1$ are anti-reflexive,
\item  either both $C_2, D_2$ are anti-reflexive or both are $=$,
\item $R$ entails $(C_1(x_1, x_2) \implies D_1(x_2, x_3))$
\item $R$ entails $(D_1(x_2, x_3) \implies C_1(x_1, x_2))$
\item $R$ contains both 
$$(R_1(x_1, x_2, x_3) \equiv (C_1(x_1, x_2) \wedge D_1(x_2, x_3)))$$
and
$$(R_2(x_1, x_2, x_3) \equiv (C_2(x_1, x_2) \wedge D_2(x_2, x_3)))$$
\end{itemize}
\end{definition}

Before we turn to the main result of this subsection we need two observations.
%(The proof of the first observation as well as other omitted proofs may be found in the appendix.)

\begin{observation}
\label{obs:TernaryCTowardsD} 
Let $\structA$ be a liberal finitely bounded homogeneous binary core and $R$ a critical ternary relation over  $(\structB, C_1, C_2, D_1, D_2)$
for some first-order expansion $\structB$ of $\structA$.
Let $k \in \N$ and $I, J \subseteq [k] \setminus \{ m \}$ for some $m \in [k]$ be disjoint subsets of indices such that
$I \cup J \cup \{ m \} = [k]
$ and $w, u \in A^k$ such that  all of the following hold:

\begin{itemize}
\item $(w[i], u[i]) \in C_1$ for all $i \in I$,
\item $(w[m], u[m]) \in C_2$, and
\item $(w[i], u[i]) \in C_2$ for all $i \in J$.
\end{itemize}
Then there exists $v \in A^k$ such that 
all of the following hold:
\begin{itemize}
\item $(w[i], u[i], v[i]) \in R_1$ for all $i \in I$,
\item $(w[m], u[m], v[m]) \in \{ (x_1, x_2, x_3) \in A^3 \mid C_2(x_1, x_2)
\wedge
D_1(x_2, x_3)$
$ \wedge x_1 \neq x_3 \}$, and 
\item $(w[i], u[i], v[i]) \in R_2$ for all $i \in J$.
\end{itemize}
\end{observation}

\begin{proof}
Starting with a substructure of the structure $\structA$ induced by the elements
$w[1], \ldots, w[k], $
$u[1], \ldots, u[k]$ satisfying
the conditions in the formulation of the lemma we will move around the structure $\structA$ to find $v[1], \ldots, v[k]$ so that the structure induced by all elements in $w, u$ and $v$  satisfy all the requirements. Assume we are done for $v[1], \ldots, v[i]$ for some $i \in [k]$ and consider $v[i+1]$. 
Then one of three requirements has to be satisfied by $v[i+1]$.
The first case to consider is where $(i+1) \in I$ and we require
$ (w[i+1], u[i+1], v[i+1]) \in R_1$. Since $\structA$ is liberal, there are three pairwise different elements $a_1, a_2, a_3$ in $\structA$ such that $(a_1, a_2)$   are in the same orbital as $(w[i+1], u[i+1])$ and $(a_2, a_3)$ is in some (actually in any) orbital contained in $D_1$. Since $\structA$ is homogeneous there is an automorphism $\alpha$  
sending $(a_1, a_2)$ to $(w[i+1], u[i+1])$. Then we take $v[i+1]$ to be $\alpha(a_3)$. When $i + 1 = m$ then either $D_2$ is anti-reflexive and we proceed similarly or $D_2$ is $=$ and  then we set $v[i+1]$ to $u[i+1]$.
If $i \in J$, then we proceed as in one of the cases above. Indeed, either $C_2, D_2 \subseteq \neq$ or both $C_2, D_2$ equal $=$.
\end{proof}

We also need to consider a situation from the previous observation where $u,v$ are given and one looks for $w$.

\begin{observation}
\label{obs:TernaryDTowardsC}
Let $\structA$ be a liberal finitely bounded homogeneous binary core  and $R$ a critical ternary relation over  $(\structB, C_1, C_2, D_1, D_2)$
for some first-order expansion $\structB$ of $\structA$.
Let $k \in \N$ and $I, J \subseteq [k] \setminus \{ m \}$ for some $m \in [k]$ be disjoint subsets of indices such that
$I \cup J \cup \{ m \} = [k]
$ and $u,v \in A^k$ such that  

\begin{itemize}
\item $(u[i], v[i]) \in D_1$ for all $j \in I$,
\item $(u[m], v[m]) \in D_2$, and
\item $(u[i], v[i]) \in D_2$ for all $i \in J$.
\end{itemize}
Then there exists $w \in A^k$ such that 
all of the following hold:
\begin{itemize}
\item $(w[j], u[j], v[j]) \in R_1$ for all $j \in I$,
\item $(w[m], u[m], v[m]) \in \{ (x_1, x_2, x_3) \in A^3 \mid C_1(x_1, x_2) \wedge D_2(x_2, x_3) \wedge x_1 \neq x_3\}$, and 
\item $(w[j], u[j], v[j]) \in R_2$ for all $j \in J$.
\end{itemize}
\end{observation}

We are now in the position to prove that critical ternary relations pp-definable in first-order expansions of liberal finitely bounded homogeneous cores have no bounded strict width.

\begin{proposition}
\label{prop:CriticalTernaryNoQNUF}
Let $\structA$ be a liberal finitely bounded homogeneous binary core  and $R$ a critical ternary relation over  $(\structB, C_1, C_2, D_1, D_2)$
for some first-order expansion $\structB$ of $\structA$.
 Then both $\structB$ and $R$ do not have bounded strict width.
\end{proposition}

\begin{proof}
By Theorem~\ref{thm:strictwidth}, it is enough to prove that $R$ is not preserved by any oligopotent qnu-operation.
Suppose $\structB$ is preserved by a $k$-ary oligopotent quasi near-unanimity operation $f$. The essence of the proof of the proposition is in the following observation. 

Let $C$ be a binary relation. We will say that two $k$-tuples $t_1, t_2$ are $C$-connected on a coordinate $i \in [k]$ if $(t_1[i], t_2[i]) \in C$. 

\begin{observation}
\label{obs:TernaryCDSliding}
For all $i \in \{ 0, \ldots, k \}$ we have that both of the following hold.
\begin{enumerate}
\item \label{ternaryeq:A} For all $(w,u) \in (A^k)^2$ such that  $u$ is constant, $(w,u)$ are $C_1$-connected on $(k - i)$ coordinates and $C_2$-connected on $i$ coordinates, we have that $(f(w), f(u)) \in C_1$.
\item \label{ternaryeq:B} For all $(u,v) \in (A^k)^2$ such that $u$ is constant, 
$(u,v)$ are $D_1$-connected on $(k - i)$ coordinates and $D_2$-connected on $i$ coordinates, we have that $(f(u), f(v)) \in D_1$.
\end{enumerate} 
\end{observation}

\begin{proof}
The proof goes by the induction on $i$. In the base case where $i = 0$, the claim follows by the fact that $C_1, D_1$ are pp-definable in $\structB$. 

For the induction step, suppose first that Item~\ref{ternaryeq:B} fails for some $i>1$ and $J$ with $\left| J  \right| =i$, i.e., there exist    $(u,v) \in (A^k)^2$ such that $u$ is constant, 
$(u,v)$ are $D_1$-connected on $(k \setminus [J])$ and $D_2$-connected on 
$J$ and $(f(u), f(v)) \notin D_1$.
 Let $m \in J$. We set $J' = J \setminus \{ m \}$  
and  $I' = [k] \setminus J$. It follows that $[k]$ is a disjoint union of $I', \{ m \}$
and $J'$ and we have all of the following:
\begin{itemize}
\item $(u[j], v[j]) \in D_1$ for all $j \in I'$,
\item $(u[m], v[m]) \in D_2$, and
\item $(u[j], v[j]) \in D_2$ for all $j \in J'$
\end{itemize}

By Observation~\ref{obs:TernaryDTowardsC} there exists $w \in A^k$ such that all of the following hold:
\begin{itemize}
\item $(w[j], u[j], v[j]) \in R_1$ for all $j \in I'$,
\item $(w[m], u[m], v[m]) \in  \{ (x_1, x_2, x_3) \in A^3 \mid C_1(x_1, x_2)\wedge $
$D_2(x_2,x_3)$
$ \wedge x_1 \neq x_3 \}$, and 
\item $(w[j], u[j], v[j]) \in R_2$ for all $j \in J'$.
\end{itemize}

By the induction hypothesis, it holds that $(f(w),f(u)) \in C_1$. 
Now, since $\structA$ is liberal  and 
$C_1, D_1 \subseteq \neq$ 
there exist pairwise different 
$b_1, b_2, b_3$ such that $(b_1, b_3)$ are in the same orbital as $(w[m], v[m])$,
$(b_1, b_2) \in C_1$ and $(b_2, b_3) \in D_1$.
Since $\structA$ is homogeneous, there is an automorphism $\alpha \in \structA$ sending $(b_1, b_3)$ to $(w[m], v[m])$. Let $a$ be such that $a = \alpha(b_2)$ and
$u' \in A^k$ such that $u'[m] = a$ and $u'[j] = u[j]$ whenever $j \neq m$.
Observe that  $(w[m], a, v[m]) \in R_1$. 

Since for all $j \in [k]$ we have that $(w[j], u'[j], v[j]) \in R_1$
or $(w[j], u'[j], v[j]) \in R_2$, it follows that $(w[j], u'[j], v[j]) \in R$
for all $j \in [k]$. The relation $R$ is preserved by $f$. Thus, 
$(f(w), f(u'), f(v)) \in R$. The tuple $u$ is constant, $f$ is an oligopotent qnu-operation, and hence $f(u') = f(u)$.  It implies that $(f(w), f(u')) \in C_1$.
The relation $R$ entails $(C_1(x_1,x_2)$
$\implies D_1(x_2, x_3))$, and hence
$(f(u'), f(v)) \in D_1$. Since $f(u') = f(u)$, it follows that $(f(u), f(v)) \in D_1$.
It contradicts the assumpution that $(f(u), f(v)) \notin D_1$ and proves that Item~\ref{ternaryeq:B} holds for the induction step.

The proof that  Item~\ref{ternaryeq:A} goes through the induction step  is analogous to the proof for Item~\ref{ternaryeq:B} with a difference that we use Observation~\ref{obs:TernaryCTowardsD} instead of Observation~\ref{obs:TernaryDTowardsC}.
% i.e., there exist  $(v,u) \in (A^k)^2$ such that  $u$ is constant, $(v,u)$ are $C_1$-connected on $(k - i)$ coordinates and $C_2$-connected on $i$ coordinates we have that and $(f(v), f(u)) \notin A_1$. 
%Since $R_1, R_2$ are non-empty and $\Delta$ is homogeneous,
%we have that there exists $w \in D^k$
%such that for all $j \in [k]$ we have that $(v[j], u[j], w[j])$ is either in $R_1$ or in $R_2$. Since Item~\ref{ternaryeq:B} holds for the induction step, we have that 
%$(f(u), f(w)) \in D_1$. The relation $R$ entails $(D_1(x_2,x_3) \to C_1(x_1, x_2))$,
%and hence $(f(u), f(v)) \in C_1$. 
It completes the proof of the induction step and the observation. 
\end{proof}

Observation~\ref{obs:TernaryCDSliding} implies $f(D_2, \ldots, D_2) = D_1$.
It contradicts the fact that $D_2$ is pp-definable in $\structB$ and completes the proof of the proposition. 
\end{proof}

\section{Efficient  Entailment and Implications}
\label{sect:EffEntImp}

In this section,   we first provide the definition of an implication, which is needed to formally define implicationally simple and implicationally hard structures. Then,
in the following subsections, we prove certain preliminary results  which are then used  in the proof of the main theorem. In particular, in Section~\ref{sect:FromImpToCriticalTernary}, we show some auxiliary results that will be used to pp-define critical ternary relations out of a pair of complementary implications.
In Section~\ref{sect:SomeImp}, we show that some implications are not pp-definable  in first-order expansions of liberal finitely bounded homogeneous binary cores with bounded strict width, in Section~\ref{sect:CompImp}  how to compose implications 
and in Section~\ref{sect:Bipartite} that out of a pair of complementary implications we can always pp-define an implication of a very concrete form called a complete implication.

We start with a definition of \emph{efficient entailment} which is a non-standard version of entailment from Section~\ref{sect:preliminaries} and concerns ternary and quaternary relations. First, we look into the ternary ones. 

\begin{definition}
\label{def:EffEntTernary}
Let $R$ be a ternary relation and $C_1,D_1$ binary relations. We say that $R$ efficiently entails:
$$(C_1(x_i, x_j) \implies D_1(x_k, x_l))$$
with $i,j,k,l \in [3]$ if  both of the following hold:
\begin{itemize}
\item $R(x_1, x_2, x_3)$ entails $(C_1(x_i, x_j) \implies D_1(x_k, x_l))$,
\item $C_1 \subsetneq \proj_{i,j} R$ and $D_1 \subsetneq \proj_{k,l} R$.
\end{itemize}
\end{definition}

\noindent
We have a similar definition for quaternary relations.

\begin{definition}
\label{def:EffEntQuaternary}
Let $R$ be a quaternary relation and $C_1,D_1$ binary relations. We say that $R$ efficiently entails:
$$(C_1(x_i, x_j) \implies D_1(x_k, x_l))$$
with $i,j,k,l \in [4]$
if  both of the following hold:
\begin{itemize}
\item $R(x_1, x_2, x_3, x_4)$ entails $(C_1(x_i, x_j) \implies D_1(x_k, x_l))$,
\item $C_1 \subsetneq \proj_{i,j} R$ and $D_1 \subsetneq \proj_{k,l} R$.
\end{itemize}
\end{definition}

\noindent
Next we define a  ternary 'implication'.

\begin{definition}
\label{def:TernaryImp}
Let $L,P \in \{ \leftarrow, \rightarrow \}$. 
We say that a ternary relation $R$ is a ternary  $(\structB, C,D,C_1, D_1, L,P)$-implication if all of the following hold: 
\begin{enumerate}
\item \label{TernaryImp:EntEq} $R$ entails no equalities,
\item \label{TernaryImp:PPDefn} all $R, C_1, D_1$ are pp-definable in $\structB$, 
\item \label{TernaryImp:Proj12} $\proj_{i,j} R = C$ where $(i,j) = (1,2)$ if $L = \rightarrow$ and $(i,j) = (2,1)$ if $L = \leftarrow$,
\item \label{TernaryImp:Proj23} $\proj_{k,l} R = D$ where $(k,l) = (2,3)$ if $L = \rightarrow$ and $(k,l) = (3,2)$ if $L = \leftarrow$,
\item \label{TernaryImp:EffEnt} $R$ efficiently entails $(C_1(x_i, x_j) \implies D_1(x_k, x_l))$.
\end{enumerate} 
\end{definition}

We now provide a similar definition for quaternary relations.

\begin{definition}
Let $L,P \in \{ \leftarrow, \rightarrow \}$. 
A quaternary relation $R$ is a quaternary $(\structB, C,D,C_1, D_1, L,P)$-implication if all of the following hold:
\begin{enumerate}
\item \label{QuaternaryImp:EntEq} $R$ entails no equalities,
\item \label{QuaternaryImp:PPDefn} all $R, C_1, D_1$ are pp-definable in $\structB$,
\item \label{QuaternaryImp:Proj12} $\proj_{i,j} R = C$ where $(i,j) = (1,2)$ if $L = \rightarrow$ and $(i,j) = (2,1)$ if $L = \leftarrow$,
\item \label{QuaternaryImp:Proj34} $\proj_{k,l} R = C$ where $(k,l) = (3,4)$ if $L = \rightarrow$ and $(k,l) = (4,3)$ if $L = \leftarrow$,
\item \label{QuaternaryImp:EffEnt} $R$ efficiently entails $(C_1(x_i, x_j) \implies D_1(x_k, x_l))$.
\end{enumerate} 
\end{definition}

For the sake of succinctness we say that  $R$ is a $(\structB, C,D,C_1, D_1, $
$L, P)$-implication
without specifying that it is ternary or quaternary if not necessary, 
or that a relation $R$ is a (ternary or quaternary)  $(\structB, C,D,C_1, D_1)$-implication
if it is a $(\structB, C,D,C_1, D_1, L,$
$ P)$-implication for some $L,P \in \{  \leftarrow, \rightarrow \}$ or that $R$ is a  (ternary or quaternary) $(L,P)$-implication if it is a 
$(\structB, C,D,C_1, D_1, $
$L, P)$-implication for some $\structB, C,D,C_1, D_1$.

\paragraph{Example.} Let the relation $R$ be a critical ternary relation over $(\structB, C_1, C_2,$
$ D_1, D_2)$. Observe that the relation $R$ is a $(\structB, \proj_{1,2} R, \proj_{2,3} R, $
$C_1, D_1,$
$ \rightarrow, \rightarrow)$-implication and that $$R'(x_1, x_2, x_3)
\equiv R(x_3, x_2, x_1)$$ is
a $(\structB, \proj_{3,2} R, \proj_{2,1} R, D_1^{-1}, C_1^{-1}, \rightarrow, \rightarrow)$-implication.

\subsection{From Implications to Critical Ternary Relations} 
\label{sect:FromImpToCriticalTernary}

In order to define critical ternary relations, for instance, out of a pair (or a bunch) of implications we provide pp-definitions in many steps using often similar constructions. For the sake of succinctness, we will use certain shorthands.

\begin{definition}
\label{defn:bowtie}
We write $R_3 := R_1 \bowtie R_2$
for a ternary relation $R_3(x_1, x_2, x_3)$ defined out of two ternary relations $R_1$ and $R_2$ as follows
\begin{align}
\label{eq:BowtieTernary}
\exists y~R_1(x_1, x_2, y) \wedge R_2(y, x_2, x_3)
\end{align}
or a quaternary relation $R_3(x_1, x_2, x_3, x_4)$  defined out of two quaternary relations $R_1, R_2$ as follows:
\begin{align}
\label{eq:BowtieQuaternary}
\exists y_1 \exists y_2~R_1(x_1, x_2, y_1, y_2) \wedge R_2(y_2, y_1, x_3, x_4)
\end{align}
\end{definition}

We also write $R_1 \bowtie_3 R_2$ for a ternary relation $R_3(x_1, x_2, x_3)$ defined out of two quaternary relations 
$R_1, R_2$ as follows:
\begin{align}
\label{eq:BowtieThree}
\exists y_1 \exists y_2~R(x_1, x_2, y_1, y_2) \wedge R(y_2, y_1, x_2, x_3)
\end{align}

The following observation provides us with some insight into the structure of  $R_1 \bowtie R_2$ on the condition that we provide some assumptions on $R_1, R_2$.

\begin{observation}
\label{obs:AllInjective}
Let $R_1, R_2$ be both ternary or both quaternary relations fo-definable in a liberal finitely bounded homogeneous binary core $\structA$, the relation $R_3 := R_1 \bowtie R_2$ and $O_1, O_2, O_3$ some orbitals. Then all of the following hold:
\begin{itemize}
\item if $R_1$ contains an $O_1 O_2$-tuple $t_1$ and $R_2$ contains an  $O_2^{-1} O_3$-tuple $t_2$ then $R_3$ contains an $O_1 O_3$-tuple, 
\item if $t_1$ and $t_2$ are injective, then $R_3$ contains all injective $O_1 O_3$-tuples
\item if $O_1, O_2, O_3$ are $=$ and both $t_1$ and $t_2$ are non-constant, then $R_3$ contains all non-constant $==$-tuples.
\end{itemize}
\end{observation}

\begin{proof}
Consider the first item in the formulation of the observation and the case where $R_1$ and $R_2$ are ternary.
Let $(a_1, a_2, a_3)$ be an $O_1 O_2$-tuple in $R_1$ and $(b_2, b_3, b_4)$ an $O_2^{-1} O_3$-tuple in $R_2$. Since $\structA$ is homogeneous, there exists
an automorphism $\alpha \in \Aut(\structA)$ sending   $(b_3, b_2)$ to $(a_2, a_3)$. Let $a_4$ be $\alpha(b_4)$. Observe that an assignment $\assignment: \{ x_1, x_2, x_3, y \} \to A$ such that $\assignment(x_1) = a_1$, $\assignment(x_2) = a_2$, $\assignment(y) = a_3$ and $\assignment(x_3) = a_4$.
Since $(a_3, a_2, a_4)$ is in the same orbit as $(b_2, b_3, b_4)$, we have that 
the assignment $\assignment$ satisfies all atomic formulae in~(\ref{eq:BowtieTernary}), we have that $R_3$ has an $O_1 O_3$-tuple. The proof for quaternary $R_1, R_2$ is similar with a difference that  we consider an $O_1 O_2$-tuple
$(a_1, a_2, a_3, a_4)$  in $R_1$, $O_2^{-1} O_3$-tuple $(b_3, b_4, b_5, b_6)$  in $R_2$, and an automorphism $\alpha \in \Aut(\structA)$ sending $(b_4, b_3)$ to $(a_3, a_4)$.

Now turn to the second item for ternary relation $R_1$ which contains an $O_1 O_2$-tuple $t_1$ and $R_2$ with an $O_2^{-1} O_3$-tuple $t_2$. Since $\structA$ is
liberal we have that for any orbital $O_{1,4}$ there exists a substructure of $\structA$ over four pairwise different elements $a_1, a_2, a_3, a_4$ such that $(a_1, a_2, a_3)$ is in the same orbit as $t_1$, $(a_3, a_2, a_4)$ in the same orbit as $t_2$ and $(a_1, a_4) \in O_{1,4}$. It follows that $R_3$ has an $O_1 O_3$-tuple $t$ such that $(t[1], t[3]) \in O_{1,4}$. Since $O_{1,4}$ was chosen arbitrarily, we have that $R_3$ contains all injective tuples. The proof for quaternary relations is again similar. We just look at a substructure of $\structA$
over six elements $a_1, a_2, a_3, a_4, a_5, a_6$ such that $(a_1, a_2, a_5, a_6)$  
is in the same orbit as $t_1$, the tuple $(a_6, a_5, a_3, a_4)$ is in the same orbit as $t_2$ and such that for all $(i,j) \in \{(1,3), (1,4), (2,3), (2,4) \}$ we have
$(a_i, a_j) \in O_{i,j}$ for some anti-reflexive orbital $O_{i,j}$. An assignment $\assignment: \{ x_1, \ldots, x_4, y_1, y_2 \} \to A$ such that $\assignment(x_i) = a_i$ for $i \in [4]$ and $\assignment(y_i) = a_{i+4}$ for $i \in [2]$ clearly satisfies both atomic formulae in~(\ref{eq:BowtieQuaternary}). Since $O_{i,j}$ with  $(i,j) \in \{(1,3), (1,4), (2,3), (2,4) \}$ were chosen arbitrarily, it follows that $R_3$ contains all injective $O_1 O_3$-tuples. It completes the proof of the second item.
  
 Since there are no ternary non-constant $==$-tuples, the last item concerns quaternary relations only. We proceed as in the proof for the second item. 
Let $t_1$ be a non-constant tuple in $R_1$ and $t_2$ a non-constant tuple in $R_2$. This time we look at a substructure of $\structA$ induced by pairwise different $a_1, a_2, a_3$ such that $(a_1, a_2)$ is in the same orbital as $(t_1[2], t_1[3])$,   $(a_2, a_3)$ in the same orbit as $(t_2[2], t_2[3])$, and $(a_1, a_3)$ in some anti-reflexive  orbital $O_{1,3}$. The assignment $\assignment$ sending $x_1, x_2$ to $a_1$, $y_1, y_2$ to $a_2$ and $x_3, x_4$ to $a_3$ satisfies both atomic formulae in~(\ref{eq:BowtieQuaternary}) and hence provides an $==$-tuple $t$ for $R$ satisfying $(t[2], t[3]) \in O_{1,3}$. Since $O_{1,3}$ was chosen arbitrarily, we have that $R_3$ contains all non-constant $==$-tuples. It completes the proof for the last item in the observation.
\end{proof}

An important part of the proof that a pp-defined ternary relation is a critical ternary relation is to show that it contains relations $(C_1(x_1, x_2) \wedge D_1(x_2, x_3))$ and 
$(C_2(x_1, x_2) \wedge D_2(x_2, x_3))$, see Definition~\ref{def:CriticalTernary}.
To this end we will use the following observation.

\begin{observation}
\label{obs:ContainsFullBowtie}
Let $R_1, R_2$ be two ternary relations with fo-definitions in a liberal finitely bounded homogeneous binary core $\structA$ such that $R_1$ contains
all injective $O_1 O_2$-tuples and $R_2$ contains all injective $O_2^{-1} O_3$-tuples for some anti-reflexive orbitals $O_1, O_2, O_3$ then $R_1 \bowtie R_2$ contains  the relation
$$O_1(x_1, x_2) \wedge O_3(x_2, x_3).$$
\end{observation}

\begin{proof}
If $O_3$ is different from $O_1^{-1}$, then the observation follows by Observation~\ref{obs:AllInjective}. Indeed, in this case it is enough to prove
that $R_3$ contains all injective $O_1 O_3$-tuples. 
Otherwise, we repeat the proof of Observation~\ref{obs:AllInjective} to show that  $R$ has all injective $O_1 O_1^{-1}$-tuples but we  have also to show that $R_3$ contains an $O_1 O_1^{-1}$ tuple $t$ with $t[1] = t[3]$.  To this end, consider any pairwise different $a_1, a_2, a_3 \in A$ such that 
$(a_1, a_2) \in O_1, (a_2, a_3) \in O_2$ and $(a_1, a_3) \in O_{1,3}$ for some anti-reflexive orbital $O_{1,3}$.  They exist since $\structA$ is liberal.
The relation  $R_1$ contains all injective 
$O_1 O_2$-tuples and $R_2$ contains all injective $O_2^{-1} O_1^{-1}$-tuples,
hence in particular $R_1$ contains $(a_1, a_2, a_3)$  and $R_2$  contains $(a_3, a_2, a_1)$. It follows that
 the assignment $\assignment: \{ x_1, x_2, x_3, y \}$ sending $x_1, x_3$ to $a_1$,  
$x_2$ to $a_2$ and $y$ to $a_3$ satisfies both atomic formulae in~(\ref{eq:BowtieTernary}) and provides an $O_1 O_1^{-1}$-tuple with $t[1] = t[3]$ in $R_3$.  
\end{proof}

The next observation will be used in a similar context as Observation~\ref{obs:ContainsFullBowtie}. The difference is that here we will look at quaternary not ternary relations $R_1, R_2$ and at $\bowtie_3$ not $\bowtie$.

\begin{observation}
\label{obs:ContainsFullBowtieThree}
Let $R_1, R_2$ be two quaternary relations with fo-definitions in a liberal finitely bounded homogeneous binary core $\structA$ such that $R_1$ contains
all injective $O_1 O_2$-tuples and $R_2$ contains all injective $O_2^{-1} O_3$-tuples or all $O_1, O_2, O_3$ are $=$ and both $R_1, R_2$ contain all non-constant $==$-tuples
then $R_1 \bowtie_3 R_2$ contains:
$$O_1(x_1, x_2) \wedge O_3(x_2, x_3).$$
\end{observation}

\begin{proof}
Since $\structA$ is liberal, we have that for anti-reflexive orbitals $O_{1,3}$ it contains a substructure over pairwise different elements 
$a_1, a_2, a_3, a_4, a_5$ such that $(a_1, a_2) \in O_1$, $(a_4, a_5) \in O_2$, $(a_2, a_3) \in O_3$ and $(a_1, a_3) \in O_{1,3}$.
Since $R_1$ contains all injective $O_1, O_2$-tuples  and $R_2$ contains all
injective $O_2^{-1} O_3$-tuples, we have that an assignment $\assignment: \{ x_1, x_2, x_3, y_1, y_2 \} \rightarrow A$ such that $\assignment(x_i) = a_i$ for $i \in [3]$ and $\assignment(y_i) = a_{i+3}$ for $i \in [2]$ satisfies both atomic formulae in~(\ref{eq:BowtieThree}). Since $O_{1,3}$ may be arbitrary, it proves that $R_3$ has all injective $O_1 O_3$-tuples. If $O_3 = O_1^{-1}$, then we also have to show that $R_3$ has an $O_1 O_1^{-1}$-tuple $t$ with $t[1] = t[3]$. The proof is similar with a difference that 
we choose $a_1, a_2, a_3, a_4, a_5$ with $a_1 = a_3$.
\end{proof}

\subsection{Some Implications with no Bounded Strict Width}
\label{sect:SomeImp}

We are now ready to prove that certain ternary and quaternary relations pp-define a critical ternary relation and hence do not have bounded strict width.

\begin{lemma}
\label{lem:EqOrEqOrEqnoQNUF}
Let $\structA$ be a liberal finitely bounded homogeneous 
binary core  and $\structB$ 
a first-order expansion of $\structA$ which
pp-defines a ternary relation  $R$ that entails no equalities but   
entails 
$$\varphi(x_1, x_2, x_3) \equiv (x_1 = x_2 \vee x_2 = x_3 \vee x_1 = x_3).$$ Then $R$, and hence $\structB$ do not have bounded strict width.
\end{lemma}

\begin{proof}
Since $R$ entails no equalities, we have that  if $R$ entails
any subformula of $\varphi$, then this subformula has at least two disjuncts.
Assume without loss of generality that in this case $R$ entails 
$(x_1 = x_2 \vee x_2 = x_3)$. Since $R$ entails no equalities we have that 
$(R'(x_1, x_2, x_3) \equiv R(x_1, x_2, x_3) \wedge x_1 \neq x_3)$ is equivalent to
$((C(x_1, x_2) \wedge x_2 = x_3) \vee (x_1 = x_2 \wedge D(x_2, x_3)))$
for some anti-reflexive  $C, D \subseteq A^2$. By Observation~\ref{obs:NeqCanBeDefined}, the relation $R'$ is pp-definable in $\structB$.

Further, since $\structA$ is liberal, it is easy to see that the relation
\begin{align}
\label{eq:Rdprime}
(R''(x_1, x_2, x_3) \equiv \exists y~R'(x_1, x_2,y) \wedge R'(x_3, x_2, y))
\end{align}
equals $(S(x_1, x_2, x_3) \equiv (C_1(x_1, x_2) \wedge C^{-1}(x_2, x_3)) \vee (x_1 = x_2 \wedge x_2 = x_3)).$ Indeed, any assignment $\assignment: \{ x_1, x_2, x_3,y \} \to A$
satisfying both atomic formulae in~(\ref{eq:Rdprime}) either sends $(x_1, x_2)$ to some pair $(a_1, a_2)$ in $C_1$
or to the same element in $A$. In the former case, we have that $\assignment(x_2) = \assignment(y)$, and hence $(\assignment(x_2), \assignment(x_3)) = (a_2, a_3)$ for some $(a_2, a_3) \in C^{-1}$
Since $\structA$ is liberal, we may choose $(a_1, a_2, a_3)$ so that $(a_1, a_3)$
is in any orbital $O_{1,3}$ if $(a_1, a_2) \in O_{1,2}$ and $(a_2, a_3)$  in $O_{2,3} = O_{1,2}^{-1}$ and so that $O_{1,3}$ is any anti-reflexive orbital if $O_{2,3} \neq O_{1,2}^{-1}$. It follows that $R''$ contains  $(R_1(x_1, x_2, x_3) \equiv (C_1(x_1, x_2) \wedge C^{-1}(x_2, x_3)))$ and that any tuple in $R''$ 
with the two first coordinates being different is in $R_1$. On the other hand, 
any assignment $\assignment: \{ x_1, x_2, x_3,y \} \to A$
satisfying both atomic formulae in~(\ref{eq:Rdprime}) 
that sends $(x_1, x_2)$ to  the same element in $A$, satisfies also 
$(\assignment(x_2), \assignment(y)) \in D$ and hence $(\assignment(x_2) = \assignment(x_3))$. It follows that $R''$ equals $S$.

Clearly the relation $S$ efficiently entails both: $(C(x_1, x_2) \implies C^{-1}(x_2, x_3))$ and
$(C^{-1}(x_2, x_3) \implies C(x_1, x_2))$ and clearly $C$ and $=$ are pp-definable in 
$\structB$.
Hence $R$ is a critical ternary relation over $(\structB, C, =, C^{-1}, =)$.  It follows by Proposition~\ref{prop:CriticalTernaryNoQNUF} that  $\mathbb{B}$ do not have bounded strict width.

Assume now that $R$ entails no subformulae of $\varphi$. Since $R$, however, entails $\varphi$, we have that also in this case $(R'(x_1, x_2, x_3) \equiv R(x_1, x_2, x_3) \wedge x_1 \neq x_3)$ is equivalent to
$((C(x_1, x_2) \wedge x_2 = x_3) \vee (x_1 = x_2 \wedge D(x_2, x_3)))$
for some anti-reflexive  $C, D \subseteq A^2$. The proof is hence identical as in the previous case. It completes the proof of the lemma. 
\end{proof}

As we will see $(\structB, C,D, O,  =)$-implications where $\structB$ is a first-order expansion  of a liberal finitely bounded homogeneous binary core  and $O$ an anti-reflexive orbital  do not have bounded strict width. We start with ternary implications.

\begin{lemma}
\label{lem:TernaryOImpliesEq}
Let $\structA$ be a liberal finitely bounded homogeneous 
binary core  and $\structB$ a first-order expansion of $\structA$, let $R$ be a ternary relation pp-definable in $\structB$ that entails no equalities and efficiently entails 
$$(O(x_1, x_2) \implies x_2 = x_3)$$ for some anti-reflexive orbital $O$.
Then $\structB$ does not have bounded strict-width.
\end{lemma}

\begin{proof}
Since $R$ entails no equalities, by Lemma~\ref{lem:EqOrEqOrEqnoQNUF},
we have that $R$ contains an injective tuple $t$ such that in particular 
$(t[1], t[2]) \in P$ for some anti-reflexive orbital $P$ different from $O$. 
Consider now the relation:
\begin{align}
\label{TernaryOImpliesEq}
R'(x_1, x_2, x_3) \equiv \exists y~R(x_1, x_2, y) \wedge O(x_3, y)
\end{align}
and observe that $R'$ 
\begin{itemize}
\item efficiently
entails 
$\eta_O:=(O(x_1, x_2) \implies O^{-1}(x_2, x_3))$,
\item contains an injective $O O^{-1}$-tuple, and
\item an injective $P P^{-1}$ tuple.
\end{itemize}

The first item follows from the fact that the relation $R$ entails $(O(x_1, x_2) \implies x_2 = x_3))$.
For the second item consider a substructure of $\structA$ induced by three pairwise different elements $a_1, a_2, a_3$ such that $(a_1, a_2) \in O$ and $(a_2, a_3) \in O^{-1}$. Since $\structA$ is liberal such $a_1, a_2, a_3$ exist. Observe that $(a_1, a_2, a_3)$ is an $OO^{-1}$-tuple in $R'$. Finally, consider a substructure of $\structA$ over four pairwise different elements $(a_1, a_2, a, a_3)$ such that $(a_1, a_2, a)$ is in the same orbit as $t$, $(a_2, a_3) \in P^{-1}$, and $(a_3, a) \in O$. Observe that  
 an assignment sending $y$ to $a$ and $x_i$ to $a_i$ for $i \in [3]$ satisfies all atomic formulae in~(\ref{TernaryOImpliesEq}). It follows that $R$ contains an injective $PP^{-1}$-tuple.

Then by Observation~\ref{obs:AllInjective} and~\ref{obs:ContainsFullBowtie},
it follows that
both the relation $R_3 := ((R' \bowtie R') \bowtie (R' \bowtie R'))$ and 
the relation $(R'_3(x_1, x_2, x_3) \equiv R_3(x_3, x_2, x_1))$ contain both 
\begin{itemize}
\item $R_1 := (O(x_1, x_2) \wedge O^{-1}(x_2, x_3))$ and 
\item $R_2 := (P(x_1, x_2) \wedge P^{-1}(x_2, x_3))$.
\end{itemize}
Clearly both $R_3$ and $R_3(x_3, x_2, x_1)$ efficiently entails $\eta_O$.
It follows that the relation $(R_4(x_1, x_2, x_3) \equiv (R_3(x_1, x_2, x_3) \wedge R'_3(x_1, x_2, x_3)))$ is a critical ternary relation over $(\structB, O, P, O^{-1},$
$ P^{-1})$ and the lemma follows by Proposition~\ref{prop:CriticalTernaryNoQNUF}.
\end{proof}

Roughly speaking, the following corollary says that 
if a ternary relation $R$ contains an $O_1O_2$ tuple
then it contains an injective $O_1P$-tuple for some orbital $P$ provided an orbital $O_1$ is anti-reflexive and
an injective $O_1 O_2$-tuple provided both orbitals $O_1$ and $O_2$
are anti-reflexive.

\begin{corollary}
\label{cor:InjTernaryTuple}
Let $\structA$ be a liberal finitely bounded homogeneous 
binary core  and $\structB$ a first-order expansion of $\structA$ with bounded strict width that pp-defines a ternary relation $R$ that entails no equalities. Then for any list of pairwise different two-element sets of indices $\{ i_1, j_1\}, \ldots, \{ i_m, j_m \}$ with $m \in [2]$ such that there exists a tuple $t \in R$ satisfying $t[i_k] \neq t[j_k]$ for all $k \in [m]$, the relation $R$ contains an injective tuple $t'$ such that 
$(t'[i_k], t'[j_k])$ is in the  same orbital as $(t[i_k], t[j_k])$ for all $k \in [m]$.
\end{corollary}

\begin{proof}
Assume the contary. Then, since $R$ entails no equalities, in the case where $m = 1$ the relation $R$ efficiently entails $(O(x_a, x_b) \implies x_c = x_d)$ for some $a,b,c,d \in [3]$. It contradicts Lemma~\ref{lem:TernaryOImpliesEq} and completes the proof in this case.

For $m = 2$ we have that $R$  entails the formula
$(O(x_a, x_b) \wedge O(x_b, x_c) \implies x_a = x_c)$ 
where $\{ a,b, c \} = \{ 1,2,3 \}$. Without loss of generality assume that $R$  entails $(O(x_1, x_2) \wedge O(x_1, x_3) \implies x_2 = x_3)$.
Observe that $R$ contains a tuple $t$ such that 
$(t[1], t[2]), (t[1], t[3]) \in O$ and $(t[2] = t[3])$.
Since 
$R$ entails neither $(O(x_1, x_2)  \implies x_2 = x_3)$
nor $(O(x_1, x_3) \implies x_2 = x_3)$. It follows that $R$ contains a tuple $t$
such that $(t[1], t[3]) \in O$ and $(t[2] \neq t[3])$. 
It follows that $(R'(x_1, x_2, x_3) \equiv R(x_1, x_2, x_3) \wedge O(x_1, x_3))$
efficiently entails $(O(x_1, x_2) \implies x_2 = x_3)$. The corollary follows by Lemma~\ref{lem:TernaryOImpliesEq}.
\end{proof}

We now turn to showing that quaternary $(\structB, C,D, O,  =)$-implications with anti-reflexive orbitals $O$ have no bounded strict width. We  first consider the case where the implication has an injective tuple.

\begin{lemma}
\label{lem:QuaternaryOimpliesEq}
Let $\structA$ be a liberal finitely bounded homogeneous  
binary core and $\structB$ a first-order expansion of $\structA$ that pp-defines a quaternary relation $R$ that efficiently entails 
$$(O(x_1, x_2) \implies x_3 = x_4)$$
and contains an injective tuple. Then $\structB$ does not have bounded strict width.
\end{lemma}

\begin{proof}
By the assumptions of the lemma, the relation $R$ contains an injective tuple $t_P$ such that 
$(t_P[1], t_P[2]) \in P$ for some anti-reflexive orbital $P$ different from $O$.
By Corollary~\ref{cor:InjTernaryTuple}, the projection of the relation $R$
to its three first arguments contains an injective tuple $t'$ with $(t'[1], t'[2]) \in O$, and hence $R$ contains an $O\!\!=$-tuple $t_O$ such that  $t_O[1], t_O[2], t_O[3]$ are pairwise different. Further, we claim that the relation
\begin{align}
\label{QuaternaryOimpliesEq}
R'(x_1, x_2, x_3, x_4) \equiv \exists y~R(x_1, x_2, x_3, y) \wedge O(x_4, y)
\end{align}
satisfies all of the following:
\begin{itemize}
\item it  efficiently entails 
$(O(x_1, x_2) \implies O^{-1}(x_3, x_4))$,
\item contains an injective $OO^{-1}$-tuple,
and 
\item an injective $PP^{-1}$-tuple.
\end{itemize}

The first item follows clearly by the fact that $R$ efficiently entails $(\eta_{\textrm{Eq}} := (O(x_1, x_2) \implies x_3 = x_4))$. For the second item, consider four pairwise different elements $a_1, a_2, a_3,a_4$ in $A$ such that $(a_1, a_2, a_3, a_3)$ is in the same orbit as $t_O$ and $(a_3, a_4) \in O^{-1}$. Since $\structA$ is liberal the elements $a_1, a_2, a_3, a_4$ exist. Observe that the assignment $\assignment: \{ x_1, \ldots, x_4, y\} \to A$ sending $x_i$ to $a_i$ for $i \in [4]$ and $y$ to $a_3$ saisfies all atomic formulae in~(\ref{QuaternaryOimpliesEq}). It follows that $R'$ has an injective  $OO^{-1}$-tuple. For the last item consider pairwise different $a_1, a_2, a_3, a, a_4 \in A$ such that $(a_1, a_2, a_3, a)$ is in the same orbit as $t_P$ and $a_4$ is such that $(a_4, a) \in O$ and $(a_3, a_4) \in P^{-1}$. Again, $a_1, a_2, a_3, a, a_4$ exist since $\structA$ is liberal. Observe that  the assignment $\assignment: \{ x_1, \ldots, x_4, y\} \to A$ sending $x_i$ to $a_i$ for $i \in [4]$ and $y$ to $a$ saisfies all atomic formulae in~(\ref{QuaternaryOimpliesEq}). Hence $R'$ contains an injective 
$PP^{-1}$-tuple. 

Then by Observation~\ref{obs:AllInjective} and~\ref{obs:ContainsFullBowtieThree}, we have that the ternary relation
$R_3 := ((R' \bowtie R') \bowtie_3 (R' \bowtie R'))$ and 
the relation $(R'_3(x_1, x_2, x_3) \equiv R_3(x_3, x_2, x_1))$ contains both 
\begin{itemize}
\item $R_1 := (O(x_1, x_2) \wedge O^{-1}(x_2, x_3))$ and 
\item $R_2 := (P(x_1, x_2) \wedge P^{-1}(x_2, x_3))$.
\end{itemize}
By the definitions of $\bowtie$ and $\bowtie_3$, both $R_3$ and $R'_3$ entail $(\eta_O \equiv (O(x_1, x_2) \implies O^{-1}(x_2, x_3)))$. Since both the relation $R_3$ and $R'_3$ contain $R_2$, we have that both $R_3$ and $R'_3$ efficiently entail $\eta_O$.
It follows that $(R_3(x_1, x_2, x_3) \wedge R_3(x_3, x_2, x_1))$ is a critical ternary relation over $(\structB, O, P, O^{-1}, P^{-1})$ and the lemma follows by Proposition~\ref{prop:CriticalTernaryNoQNUF}.
\end{proof}

Next, we look at $(\structB, C,D, O,  =)$-implications $R$ where $O$ is an anti-reflexive orbital and $R$ does not have an injective tuple.

\begin{lemma}
\label{lem:EqOrEqQuaternary} 
Let $\structA$ be a liberal finitely bounded homogeneous binary core and $\structB$ a first-order expansion of $\structA$ that pp-defines a quaternary relation $R$ that  entails 
$$(x_1 = x_2 \vee x_3 = x_4)$$
and entails no equalities.
Then $\structB$ does not have bounded strict width.
\end{lemma}

\begin{proof}
Observe that every binary core with one injective orbital only is isomorphic to $(\N; \neq, =)$. By Theorem~1 in~\cite{CPWrona12}, see also the introduction in this paper, it follows that every relation $R$ that has both a first-order definition in $(\N; \neq, =)$ and
bounded strict width  is definable by a conjunction of disjunctions of disequalities, i.e., conjunctions of clauses of the form $(x_1 \neq y_1 \vee \cdots \vee x_k \neq y_k)$. Such relations $R$ clearly do not efficiently entail $(x_1 = x_2 \vee x_3 = x_4)$, and hence we can assume that $\structA$ contains at least two different anti-reflexive orbitals. This assumption will be in effect in the remainder of the proof of the lemma.

Let now $C, D \subseteq \neq$ be two anti-reflexive
binary relations such that every  $O\!\!=$-tuple in $R$ with an antireflexive $O$
satisfies $O \in C$ and every
$=\!\!\!O$-tuple  in $R$ with an antireflexive $O$ satisfies $O \in D$.
We set $O$ to be some orbital in $C$ and $P$ to be some 
anti-reflexive orbital different from $O$. We claim that the relation
$$R'(x_1, x_2, x_3, x_4) \equiv \exists y~R(x_1, x_2, x_3, y) \wedge O(y, x_4),$$
satisfies all of the following:
\begin{itemize}
\item $R'$  efficiently entails $(P(x_3, x_4) \implies x_1 = x_2)$
\item $R'$ contains an injective tuple.
\end{itemize}

For the first item consider an assignment $\assignment: \{ x_1, x_2, x_3, x_4, y \}$ sending $(x_3, x_4)$ to $P$. Since $O$ is different from $P$, we have $\assignment(x_3) \neq \assignment(y)$. It implies $\assignment(x_1) = \assignment(x_2)$. Hence $R'$ entails $(P(x_3, x_4) \implies x_1 = x_2)$.
Further, consider $(a_1, a_2, a_3, a, a_4)$ such that $(a_1, a_2, a_3, a) \in R$
and $a_3 \neq a$, $(a_3, a_4) \in P$, and $(a, a_4) \in O$. Since $\structA$ is liberal and $R$ entails no equalities, such a tuple exists. It completes the proof of the first item.

For the second item consider $(a_1, a_2, a_3, a, a_4)$ 
such that $(a_1, a_2, a_3, a) \in R$, $a_3 = a$, and
$(a_3, a_4) \in O$, and $(a_1,a_2) \in O'$ for some anti-reflexive orbital $O'$. 
We will show that $(a_1, a_2, a_3, a_4)$ is an injective $O'O$-tuple in $R'$ or $R$ has no bounded strict width. 
If the former is not the case, then $R'$ 
entails 
$(O'(x_1, x_2) \wedge O(x_3, x_4) \implies x_k = x_l)$
for some $k,l \in [4]^2 \setminus \{ \{ 1,2 \}, \{3,4\} \}$. Since $R$ entails no equalities,  $\{ k, l\} \cap \{1,2 \} \neq \emptyset$  and $\{ k,l \} \cap \{ 3,4 \} \neq \emptyset$, we have by Lemma~\ref{lem:TernaryOImpliesEq} that $R'$ either has no bounded strict width and we are done or entails neither  
$O'(x_1, x_2) \implies x_k = x_l$ nor
$O(x_3, x_4) \implies x_k = x_l$. It follows that $(R''(x_1,x_2, x_3, x_4) \wedge O(x_3, x_4))$
efficiently entails $(O'(x_1, x_2) \implies x_k = x_l)$. It follows by Lemma~\ref{lem:TernaryOImpliesEq} that $R$ has no bounded strict width and completes the proof of the second item.

Now, it is easy to see that Lemma~\ref{lem:QuaternaryOimpliesEq} applied to $\structA$, $\structB$ and the relation 
$(R''(x_1, x_2, x_3, x_4) \equiv R'(x_3, x_4, x_1, x_2))$
completes the proof of this lemma.
\end{proof}

Practically, the following corollary says that for any non-constant tuple $t$ of a quaternary relation under consideration one can find an injective tuple $t'$ that agrees with $t$ on the 'injective part'.

\begin{corollary}
\label{cor:InjQuaternaryTuple}
Let $\structA$ be a liberal finitely bounded homogeneous 
binary core and $\structB$ a first-order expansion of $\structA$ with bounded strict width that pp-defines a quaternary relation $R$ that entails no equalities. Then for any list of pairwise different two-element sets of indices $\{ i_1, j_1\}, \ldots, \{ i_m, j_m \}$ with $i_k, j_k \in [4]$  for all $k \in [m]$ such that there exists a tuple $t \in R$ satisfying $t[i_k] \neq t[j_k]$ for all $k \in [m]$, the relation contains an injective tuple $t'$ such that 
$(t'[i_k], t'[j_k])$ is in the  same orbital as $(t[i_k], t[j_k])$ for all $k \in [m]$.
\end{corollary}

\begin{proof}
Assume the contrary and let $O_{i_k,j_k}$ with $k \in [m]$ be an orbital of $t[i_k, j_k]$. Then there exists a minimal subset 	$\indices \subseteq \{ \{ i_1, j_1  \}, \ldots \{  i_m, j_m  \}  \}$ such that the formula
$$(R(x_1, x_2, x_3, x_4) \wedge \bigwedge_{\{i,j \} \in \indices} O_{i,j}(x_j, x_j))$$
entails $y_1 = y_2$ for some $y_1, y_2 \in \{ x_1, x_2, x_3, x_4 \}$.
Assume without loss of generality that $(y_1, y_2) = (x_3, x_4)$.
Let now $\indices' := \indices \setminus \{ i_0, j_0 \}$
for some $\{ i_0, j_0 \}$ in $\indices$. Since $R$ entails no equalities, the pair 
$\{ i_0, j_0 \}$ exists.
Observe that the new formula
$$(R'(x_1, x_2, x_3, x_4) \equiv R(x_1, x_2, x_3, x_4) \wedge \bigwedge_{\{i,j \} \in \indices'} O_{i,j}(x_i, x_j))$$
efficiently entails $(O_{i_0,j_0}(x_{i_0}, x_{j_0}) \implies x_3 = x_4)$.
If there is such a choice of $\indices'$ and $i_0, j_0$ that $\{ i_0, j_0 \} \cap \{ 3,4 \} \neq \emptyset$, then the corollary follows by Lemma~\ref{lem:TernaryOImpliesEq}. Otherwise $\{ i_0 ,j_0 \} = \{ 1, 2 \}$
 and $R'$ contains a tuple $t$ such that $t_1[1], t_1[2], t_1[3]$ are pairwise different and $t_1[3] = t_1[4]$ as well as a tuple $t_2$ such that $t_2[2], t_2[3], t_1[4]$ are pairwise different. If we can find $t_2$ which is injective, then we are done by Lemma~\ref{lem:QuaternaryOimpliesEq}. Otherwise $R'$ entails $(x_1 = x_2 \vee x_3 = x_4)$. Since $R'$ entails no equalities, the corollary follows by Lemma~\ref{lem:EqOrEqQuaternary}.
\end{proof}

\subsection{Composing Implications}
\label{sect:CompImp}

We will now define the way the 'implications' can be composed. We will be composing the implications originating from a 'cyclic'  instance of $\csp(\structB)$ for implicationally hard $\structB$ in order to obtain a critical ternary relation.

\begin{definition}
\label{def:Circ}
Let the relation $R_1$ be a $(\structB, C, D, C_1, D_1, L_1,$
$P_1)$-implication and 
$R_2$ a $(\structB, D, F, D_1, F_1, L_2, P_2)$-implication.
We write $R_3 := R_1 \circ R_2$ for 
a relation obtained in one of the following ways:
\begin{itemize}
\item if both $R_1$ and $R_2$ are quaternary implications and $P_1 = L_2$, then $R_3$ is quaternary and $R_3(x_1, x_2, x_3, x_4)$ is defined by:
\begin{align}
\label{Circ:44same}
\exists y_1 \exists y_2~R_1(x_1, x_2, y_1, y_2) \wedge R_2( y_1, y_2, x_3, x_4),
\end{align} 
\item if both $R_1$ and $R_2$ are quaternary implications and $P_1 \neq L_2$, then $R_3$ is quaternary and $R_3(x_1, x_2, x_3, x_4)$ is defined by:
\begin{align}
\label{Circ:44different}
\exists y_1 \exists y_2~R_1(x_1, x_2, y_1, y_2) \wedge R_2( y_2, y_1, x_3, x_4),
\end{align}
 \item if the relation $R_1$ is a quaternary $(L_1, P_1)$-implication and the relation $R_2$ a ternary $(L_2, P_2)$-implication such that $P_1 = L_2$, then 
$R_3$ is quaternary and 
 $R_3(x_1, x_2, x_3, x_4)$ is defined by 
\begin{align}
\label{Circ:43same}
\exists y~R_1(x_1, x_2, y, x_3) \wedge R_2(y, x_3, x_4)
\end{align}
 \item if the relation
 $R_1$ is a quaternary $(L_1, P_1)$-implication and the relation $R_2$ a ternary $(L_2, P_2)$-implication such that $P_1 \neq L_2$, then 
$R_3$ is quaternary and 
 $R_3(x_1, x_2, x_3, x_4)$ is defined by 
\begin{align}
\label{Circ:43different}
\exists y~R_1(x_1, x_2, x_3, y) \wedge R_2(y,  x_3, x_4)
\end{align}
 \item if the relation $R_1$ is a ternary $(L_1, P_1)$-implication and $R_2$ is a quaternary $(L_2, P_2)$-implication such that $(P_1 = L_2)$, then 
$R_3$ is quaternary and 
 $R_3(x_1, x_2, x_3, x_4)$ is defined by 
\begin{align}
\label{Circ:34same}
\exists y~R_1(x_1, x_2, y) \wedge R_2(x_2, y, x_3, x_4)
\end{align}
 \item if the relation $R_1$ is a ternary $(L_1, P_1)$-implication and $R_2$ is a quaternary $(L_2, P_2)$-implication such that $(P_1 \neq L_2)$, then 
$R_3$ is quaternary and 
 $R_3(x_1, x_2, x_3, x_4)$ is defined by 
\begin{align}
\label{Circ:34different}
\exists y~R_1(x_1, x_2, y) \wedge R_2( y,  x_2, x_3, x_4)
\end{align}
\item  if the relation $R_1$ is a ternary $(L_1, P_1)$-implication and the relation $R_2$ a ternary $(L_2, P_2)$-implication with $P_1 = L_1$, then 
$R_3$ is quaternary and 
 $R_3(x_1, x_2, x_3, x_4)$ is defined by 
\begin{align}
\label{Circ:33same}
R_1(x_1, x_2,  x_3) \wedge R_2(x_2, x_3,  x_4)
\end{align}
\item  if the relation $R_1$ is a ternary $(L_1, P_1)$-implication and the relation $R_2$ a ternary $(L_2, P_2)$-implication with $P_1 \neq L_2$, then 
$R_3$ is ternary and 
 $R_3(x_1, x_2, x_3, x_4)$ is defined by 
\begin{align}
\label{Circ:33different}
\exists y~R_1(x_1, x_2,  y) \wedge R_2(y,  x_2, x_3).
\end{align}
\end{itemize}
\end{definition}

\noindent
We will also write $(R_1)^{\circ k}$ for $(\underbrace{R_1 \circ \cdots \circ R_1}_{k \textrm{ times}})$.

We will now prove that $\circ$ states for the composition of implications, i.e., that this operation also returns an implication.

\begin{lemma}
\label{lem:ImplicationsComposition}
Let the relation $R_1$  be a $(\structB, C, D, C_1, D_1, L_1, P_1)$-implication and $R_2$ 
a $(\structB, D, F, D_1, F_1, L_2, P_2)$-implication such that both have a first-order definition in a liberal finitely bounded homogeneous structure $\structA$. 
Then the relation $R_3 := R_1 \circ R_2$ is a
$(\structB, C,F,C_1, F_1, L_1, P_2)$-implication such that $R_3$ has an $O_1 O_3$-tuple for some orbitals $O_1$ and $O_3$ if 
\begin{itemize}
\item either $P_1 = L_2$, $R_1$ has an $O_1 O_2$-tuple and $R_2$ has an $O_2 O_3$-tuple for some orbital $O_2$, or
\item $P_1 \neq L_2$, $R_1$ has an $O_1 O_2$-tuple and $R_2$ has an $O_2^{-1} O_3$-tuple for some orbital $O_2$.
\end{itemize}
\end{lemma}

\begin{proof}
Consider first the case where both $R_1, R_2$ are quaternary implications such that 
$L_2 = P_1$ and let $(a_1, a_2, a_3, a_4)$ and $(b_3, b_4, b_5, b_6)$ be an $O_1 O_2$-tuple in $R_1$ and an $O_2 O_3$-tuple in $R_2$ for some orbitals $O_1, O_2, O_3$, respectively.
Since $\structA$ is homogeneous, we have that there exists an automorphism $\alpha$ of $\structA$ sending $(a_3, a_4)$ to $(b_3, b_4)$. Let $(b_1, b_2) = (\alpha(a_1), \alpha(a_2))$. It is now easy to see that an assignment
$\assignment(x_1) = b_1$, $\assignment(x_2) = b_2$, $\assignment(y_1) = b_3$,
 $\assignment(y_2) = b_4$, $\assignment(x_3) = b_5$, $\assignment(x_4) = b_6$ satisfies both atomic formulae in~(\ref{Circ:44same}) and since $(b_5, b_6)\in O_3$, it provides the desired $O_1 O_3$-tuple in $R_3$.
Since $\proj_{3,4} R_1 = \proj_{1,2} R_2$, we have that $\proj_{1,2} R_3 = \proj_{1,2} R_1$ and $\proj_{3,4} R_2 = \proj_{3,4} R_3$.
Further,  $R_1(x_1, x_2, y_1, y_2)$ entails  $(C_1(x_i, x_j) \implies D_1(y_k, y_l))$ 
and $R_2(y_1, y_2, x_3, x_4)$ entails $(D_1(y_k, y_l) \implies F_1(x_m, x_n))$ where $i,j,k,l \in [2]$ and $m,n \in \{3,4\}$. It follows that 
$R_3$ entails $(\eta(x_1, x_2, x_3, x_4) \equiv (C_1(x_i, x_j) \implies F_1(x_k, x_l)))$. Since $C_1 \subsetneq \proj_{1,2} R_3 = \proj_{1,2} R_1$ and $F_1 \subsetneq \proj_{3,4} R_3 = \proj_{3,4} R_2$, it follows that $\eta$ is efficiently entailed by $R_3$. In order to complete the proof for this case, it is enough to prove that $R_3$ entails no equalities. Indeed, since $R_1$ entails no equalities, it contains an $O_1 O_2$-tuple where $O_1$ is anti-reflexive. By Corollary~\ref{cor:InjQuaternaryTuple}, it follows that $R_3$ contains an injective tuple $t'$ with $(t'[1], t'[2]) \in O_1$.
It implies that $R_3$ implies no equalities, and hence that
$R_3 := R_1 \circ R_2$ is a 
$(\structB, C,F,C_1, F_1, L_1, P_2)$-implication.

The proof in the case where both $R_1$ and $R_2$ are quaternary implications but 
$L_2 \neq P_1$ is similar with a difference that we look at an $O_1 O_2$-tuple $(a_1, a_2, a_3, a_4)$ and an $O_2^{-1} O_3$-tuple $(b_3, b_4, b_5, b_6)$ 
and search for an automorphism $\alpha$ that sends $(a_3, a_4)$ to $(b_4, b_3)$. 

Further, the proof in all other cases in analogous to either of the cases considered above. In the last case where $R_3$ is a ternary implication, we use Corollary~\ref{cor:InjTernaryTuple} instead of~\ref{cor:InjQuaternaryTuple}.  
\end{proof}

\subsection{Bipartite Digraph of Implications}
\label{sect:Bipartite}

In this subsection we introduce a bipartite graph that reflects the structure of 
$OP$-tuples inside a pair of ternary or quaternary implications. Observe that the relations $(C_i(x_1, x_2) \wedge D_i(x_2, x_3))$ with $i \in [2]$ definable in liberal finitely bounded homogeneous binary cores, which are  important ingredients of critical ternary relations, have all possible $OP$-tuples such that $O \subseteq C_i$ and $P \subseteq D_i$.

Consider a finitely bounded  homogeneous structure $\structA$ and $C \subseteq A^2$. We will write $\vertices_L(C)$ and $\vertices_R(C)$ 
for the set $\{  O_L \mid O \textrm{ is an orbital contained in  } C \}$  
and the set $\{  O_R \mid O \textrm{ is an orbital contained in  } C \}$, respectively.
To keep things simple, we say that a 
$(\structB, C, C, C_1,  C_1, L, P)$-implication 
is  a $(\structB, C, C_1, L, P)$-implication.

\begin{definition}
Let the relation $R_1$ and the relation $R_2$ be both $(\structB, C, C_1, \rightarrow,$
$ \leftarrow)$-implications. We define  a bipartite directed graph 
$\bipartite_{R_1, R_2}$ over two disjoint sets of vertices $\vertices_L(C)$ and $\vertices_R(C)$. 
The digraph $\bipartite_{R_1, R_2}$ contains 
\begin{itemize}
\item an arc $(O_L, P_R) \in \vertices_L(C) \times \vertices_R(C)$ if 
the relation $R_1$ contains a tuple   $(a_1, a_2, a_3, a_4)$
such that $(a_1, a_2) \in O$ and $(a_3, a_4) \in P^{-1}$ and
\item an arc $( P_R, O_L)   \in \vertices_R(C) \times \vertices_L(C)$ if 
the relation $R_2$ contains a tuple   $(a_1, a_2, a_3, a_4)$
such that $(a_1, a_2) \in P$ and $(a_3, a_4) \in O^{-1}$.
\end{itemize}

We say that $\{ O_L , P_R \}$ is a symmetric edge in $\bipartite_{R_1, R_2}$ if
it contains both an arc $(O_L,P_R)$ and $(P_R, O_L)$.
\end{definition}

We say that a subset $S$ of the vertices of a digraph is a \emph{strongly connected component} if it is a maximal set of vertices such that for any two vertices  $u, v \in S$
there is a path connecting $u$ and $v$. A set of vertices $S$
 is a set of strongly connected components if it can be partitioned into $S_1, \ldots, S_k$ so that every $S_i$ with $i \in [k]$ is a strongly connected component. We say that a (set of) strongly connected components is a \emph{sink} if every arc originating in $S$ finishes in $S$ and that $S$ is a \emph{source} if every arc finishing in $S$ also originates in $S$.

\begin{lemma}
\label{lem:TwoComponents}
Let $R_1, R_2$ be $(\structB, C, C_1, \rightarrow, \leftarrow)$-implications, $\structB$ a first-order expansion of a liberal finitely bounded homogeneous binary core $\structA$. Then  $(\vertices_L(C_1) \cup \vertices_R(C_1))$ is a set of strongly connected components which is a sink in $\bipartite_{R_1, R_2}$. 

Moreover, there exist   non-empty $D_1, F_1 \subseteq C \setminus C_1$ such that $(\vertices_L(D_1) \cup \vertices_R(F_1))$ is a source in $\bipartite_{R_1, R_2}$.
\end{lemma}

\begin{proof}
For the first part of the lemma observe that  
every arc originating in some vertex in $(\vertices_L(C_1) \cup \vertices_R(C_1))$
finishes in a vertex in $(\vertices_L(C_1) \cup \vertices_R(C_1))$. Indeed, it follows by the fact that both $R_1, R_2$ are $(\structB, C, C_1, \rightarrow, \leftarrow)$-implications.

Observe that the second part of the lemma follows by the facts that the component
$(\vertices_L(C_1) \cup \vertices_R(C_1))$ is a sink in $\bipartite_{R_1, R_2}$,
the fact that $C_1 \subsetneq C$ and that $\bipartite_{R_1, R_2}$ is finite and smooth, i.e., has no sources and no sinks. 
The graph $\bipartite_{R_1, R_2}$ is smooth since 
$\proj_{1,2} R_1 = \proj_{1,2} R_2 = \proj_{4,3}  R_1 = \proj_{4,3} R_2$. 
\end{proof}

The following lemma may be simply proved by induction using Lemma~\ref{lem:ImplicationsComposition}.

\begin{lemma}
\label{lem:CriticalPairPath}
Let $R_1, R_2$ be $(\structB, C, C_1, \rightarrow, \leftarrow)$-implications, $\structB$ a first-order expansion of a liberal finitely bounded homogeneous binary core $\structA$ and
$O_L, P_R$ a pair of vertices in $\bipartite_{R_1, R_2}$ such that there is a path from $O_L$ to $P_R$ of length $2k + 1$ for some $k \in \N$.
Then $(R_1 \circ R_2)^{\circ k} \circ R_1$  is a $(\structB, C, C_1, \rightarrow, \leftarrow)$-implication and
has an $OP^{-1}$-tuple.
\end{lemma}

%\begin{proof}
%May be simple proved by the induction on $k$ using Lemma~\ref{lem:ImplicationsEntailment}.
%\end{proof}

We will say that the relation $R$ is a complete $(\structB, C, C_1, \rightarrow, \leftarrow)$-implication if every strongly connected component in $\bipartite_{R, R}$ satisfies both of the following:
\begin{itemize}
\item  it is  of the form
$\vertices_L(D_1) \cup \vertices_R(D_1)$ for some $D_1 \subseteq C$,
\item every strongly connected component of $\bipartite_{R, R}$ is a complete bipartite digraph.
\end{itemize}

\paragraph{Example.} Let  $\structB$ be a first-order expansion of a liberal finitely bounded homogeneous binary core and $R$ a critical ternary relation over $(\structB, C_1,  D_1, $
$C_1^{-1}, D_1^{-1})$ given by $$(C_1(x_1, x_2) \wedge C_1^{-1}(x_2, x_3)) \vee(D_1(x_1, x_2) \wedge D_1^{-1}(x_2, x_3)).$$ Observe that the relation $R$ is a complete 
$(\structB, C_1 \cup D_1, C_1, \rightarrow, \leftarrow)$-implication.

We finally prove that we can always define a complete implication given two appropriate implications. 

\begin{lemma}
\label{lem:CompleteImplication}
Let $R_1, R_2$ be both  
$(\structB, C, C_1, \rightarrow, \leftarrow)$-implications where $\structB$ is a first-order expansion of a liberal finitely bounded homogeneous binary core $\structA$. Then they pp-define  a complete $(\structB, C, C_1, \rightarrow, \leftarrow)$-implication.
\end{lemma}

\begin{proof}
We say that two vertices in a bipartite digraph 
are loosely connected if  they are in the same strongly connected component of  the digraph and the shortest cycle they are both involved in is of length strictly greater than $2$. 

Consider now two $(\structB, C, C_1, \rightarrow, \leftarrow)$-implications
$R_1, R_2$ such that  not every strongly connected component of $\bipartite_{R_1, R_2}$ is complete. Then there are loosely connected vertices $O_L, P_R$ in $\bipartite_{R_1, R_1}$ such that there is an arc $(P_R, O_L)$ in 
$\bipartite_{R_1, R_2}$  but the shortest path from $O_L$ to $P_R$ is of length $2k+1$ with $k \geq 1$. By Lemma~\ref{lem:CriticalPairPath}, we have that $R_3 := (R_1 \circ R_2)^{\circ k} \circ R_1$ contains an $OP^{-1}$-tuple as well as all
$O' P'^{-1}$-tuples such that $\{ O', P' \}$ is a symmetric edge in $\bipartite_{R_1, R_2}$. Since $R_3$ is also a  $(\structB, C, C_1, \rightarrow, \leftarrow)$-implication, we have that the number of loosely connected vertices in 
$\bipartite_{R_3, R_2}$ is strictly less than in $\bipartite_{R_1, R_2}$. Since the graphs under consideration are finite, it is enough to repeat the whole procedure a finite number of times in order to obtain a pair of implications $R', R''$ without losely connected vertices.

In order to complete the proof of the lemma observe that $R' \circ R''$ is the desired complete $(\structB, C, C_1, \rightarrow, \leftarrow)$-implication. 
\end{proof}

\section{Implicationally Simple Languages}

\label{sect:ImplicationallySimple}

In this section we characterize the relational width of a large class of constraint languages which we call implicationally simple. We start with a precise definition of $\graphinstance$.

\paragraph{Graph of a CSP-instance.} Let $\instance$ be a $(2, \maxbound_{\structA})$-minimal  instance
of $\csp(\structB)$ over variables $\V$ of which we assume that $\instance$ entails no equalities. We define the implication graph $\graphinstance$ of
$\instance$  to be a directed graph over vertices which are triples of the form $((v_1, v_2), C)$ where $v_1, v_2 \in \V$ and $C \subsetneq \mathcal{I}_{v_1,v_2}$ is a binary non-empty relation pp-definable in $\structB$. There is an arc in one of the two following situations:
\begin{itemize}
\item an arc $((x_1, x_2), C_1),((x_2, x_3), D_1))$ with pairwise different $x_1, x_2, x_3 \in \V$  if there exists a constraint $\constraint \in \instance$ whose scope contains $\{ x_1, x_2, x_3 \}$,
\begin{itemize}
\item $\proj_{x_1, x_2, x_3} \constraint = ((x_1, x_2, x_3), R')$ and
\item   $R'$ is a ternary 
$(\structB, \instance_{x_1, x_2}, C_1, \instance_{x_2, x_3}, D_1 )$-implication,
\end{itemize}

\item an arc $((x_1, x_2), C_1),((x_3, x_4), D_1))$ with paiwise different $x_1, x_2, x_3, x_4 \in \V$ if there exists a constraint $\constraint \in \instance$ whose scope contains $\{ x_1, x_2, x_3, x_4 \}$ and 
\begin{itemize}
\item $\proj_{x_1, x_2, x_3, x_4} \constraint = ((x_1, x_2, x_3, x_4), R')$ and 
\item  the relation $R'$ is a quaternary 
$(\structB, \instance_{x_1, x_2}, C_1, \instance_{x_3, x_4}, $
$D_1)$-implication.
\end{itemize}
\end{itemize}

\noindent
We are now ready to define a new class of languages.

\begin{definition}
\label{defn:impsimple}
Let $\structB$ be a first-order expansion of a finitely bounded homogeneous binary core $\structA$. 
We say that a structure $\structB$ is \emph{implicationally simple} if for every $(2, \maxbound_{\structA})$ instance $\instance$ of $\csp(\structB)$
the graph $\graphinstance$ is acyclic.

If $\structB$ is not implicationally simple, then we say that it is \emph{implicationally hard}. 
\end{definition}

Naturally, every acyclic $\graphinstance$ contains a sink which is a singleton $\{ ((v_1, v_2), C \})$
for some $v_1, v_2 \in \V$ and $C \subsetneq A^2$. We will now show that in this case an instance 
$\instance[(v_1,v_2) := C]$ obtained from $\instance$ by narrowing 
down the projection of every constraint in $\instance$ containing $(v_1, v_2)$ in its scope to $C$ has the good properties of $\instance$.

\begin{observation}
\label{obs:reduction}
Let $\instance$ be a non-trivial $(2, \maxbound_{\structA})$-minimal instance of  $\csp(\structB)$ where $\structB$ is a first-order expansion of a finitely bounded homogeneous binary core $\structA$ such that $\{ ((v_1, v_2), C) \}$ is a sink in $\graphinstance$.
Then the instance $\instance[(v_1,v_2) := C]$
is also a non-trivial $(2, \maxbound_{\structA})$-minimal instance of $\csp(\structB)$.
\end{observation}

\begin{proof}
Since $\instance$ is $(2, \maxbound_{\structA})$-minimal, we have that $\instance[(v_1,v_2) := C]$ is non-trivial. To see that 
$\instance[(v_1,v_2) := C]$ is $(2, \maxbound_{\structA})$-minimal consider any constraint $\constraint = ((x_1, \ldots, x_r), R)$ in $\instance$ and any pair of variables $v_3, v_4 \subseteq \{ x_1, \ldots, x_r \}$. We have to show that for every orbital $O \subseteq \instance_{v_3, v_4}$ there is an assignment $\assignment: \{ x_1, \ldots, x_r \} \to A$ satisfying $R(x_1, \ldots, x_r)$ such that $(\assignment(v_3), \assignment(v_4)) \in O$. If $\{ v_1, v_2 \} \nsubseteq \{ x_1, \ldots, x_r \}$, then we are done by the fact that $\instance$ is $(2, \maxbound_{\structA})$-minimal. In the other case, since 
$\{ ((v_1, v_2), C) \}$ is a sink
 in $\graphinstance$,
the projection of $R(x_1, \ldots, x_r)$ to  $\{  v_1, v_2, v_3, v_4 \}$ does not efficiently entail 
$(C(v_1, v_2) \to D(v_3, v_4))$ for any $D \subsetneq I_{v_3, v_4}$. Hence there exists an assignment
$\assignment: \{ x_1, \ldots, x_r \} \to A$ satisfying $R(x_1, \ldots, x_r)$ such that $(\assignment(v_3), \assignment(v_4)) \in O$ and $(\assignment(v_1), \assignment(v_2)) \in C$. It completes the proof of the observation.
\end{proof}

The above observation is the key to showing that 
every non-trivial $(2, \maxbound_{\structA})$-minimal instance of $\csp(\structB)$ with  implicationally simple $\structB$  has a solution.

\begin{theorem}
\label{thm:ImpSimpleMinimal}
Let $\structA$ be a  finitely-bounded homogenous binary core 
and $\structB$ a first-order expansion of $\structA$ which is  implicationally simple. Then every non-trivial $(2, \maxbound_{\structA})$-minimal instance has a solution. 
\end{theorem}

\begin{proof}
Indeed,  either  there exist $v_1,v_2 \in \V$ 
such that $I_{v_1,v_2}$ consists of at least two orbitals or all $I_{v_1,v_2}$ consist of exactly one orbital $O$.
In the former case, since all  
binary relations pp-definable in $\structB$ contain an orbital,  
the graph $\graphinstance$ contains at least one vertex $((v_1, v_2), C)$
for $v_1, v_2 \in \V$ such that $\instance_{v_1,v_2}$ consists of at least two orbits.
Since $\structB$ is implicationally simple, the graph $\graphinstance$ is acyclic. In particular, it has a sink $\{ ((v_1, v_2), C) \}$ 
for some $C$. Hence, we can use Observation~\ref{obs:reduction}
and simplify the considered instance by  replacing $\instance$
with $\instance[(v_1,v_2) := C]$. The new instance is also non-trivial and 
$(2, \maxbound_{\structA})$-minimal. The process of simplifying the instance terminates when every $\instance_{v_1,v_2}$ with $v_1,v_2 \in \V$ consists of one orbit only. Clearly, every solution to this simplified instance is a solution to the original one.

From now we assume that $\instance$ is $(2, \maxbound_{\structA})$-minimal,
 $\V = \{ v_1, \ldots, v_n \}$ and for all $i,j \in [n]$
we have that $\instance_{v_i, v_j}$ consists of exactly one orbital. 
Consider a structure $\Delta$ over the signature $\tau$ of $\structA$ whose elements are variables in $\V$ and $(v_i, v_j) \in R^{\Delta}$ with $R \in \tau$ if and only if $\instance_{v_i, v_j} = R^{\structA}$. Observe first that   $E$ equal to  $=^{\Delta} \cup \{ (v_i, v_i) \mid i \in [n] \}$ is an equivalence relation. It is clearly
reflexive and symmetric. Suppose $E$ is not transitive. Then there are $v_i, v_j , v_k \in V$ such that $(v_i, v_j) \in =^{\Delta}$ and $(v_j, v_k) \in =^{\Delta}$ but $v_i, v_k \notin =^{\Delta}$. 
It implies that a constraint $\constraint$ whose scope contains $\{ v_i, v_j, v_k \}$, which is in $\instance$ by $(2, \maxbound_{\structA})$-minimality, has an empty relation. It contradicts the fact that $\instance$ is non-trivial.

In order to complete the proof of the theorem we will show that $\Delta / E$ admits an embedding to $\structA$. Assume it is not the case. Then there exists a finite structure $\Gamma \in \bounds_{\structA}$ that embeds into $\structA$. Since the size of $\Gamma$ is at most $\maxbound_{\structA}$, it contradicts the fact that $\instance$ is $(2, \maxbound_{\structA})$-minimal and non-trivial. It follows that $\Delta / E$ embeds and $\Delta$ homomorphically maps into $\structA$. It follows that $\instance$ has a solution.
\end{proof}

\section{FO-Expansions of Liberal  Finitely Bound\-ed Homogeneous Binary Cores without Bounded Strict Width}

\label{sect:LiberalBinaryCores}

In this section we show that if a first-order expansion of a liberal finitely bounded homogeneous binary core 
is implicationally hard, then it has no bounded strict width. As a result, we obtain that every first-order expansion of a liberal finitely bounded homogeneous binary core with bounded strict width is implicationally simple and hence its relational width is characterized by Theorem~\ref{thm:ImpSimpleMinimal}.

%Recall that a structure  $\structB$ is 
%implicationally hard if there exists an instance $\instance$ such that 
%$\graphinstance$ has a cycle.

\begin{lemma}
\label{lem:CycleImplication}
Let $\structA$ be a liberal finitely bounded homogeneous binary core and $\structB$ a first-order expansion of $\structA$ which
is implicationally hard. Then $\structB$ pp-defines 
\begin{itemize}
\item a quaternary 
$(\structB, C, C_1, \rightarrow, \leftarrow)$-implication, or
\item  a ternary
$(\structB, C, C_1, \rightarrow, \leftarrow)$-implication.
% \item a ternary
%$(\structB, C, C_1, \leftarrow \rightarrow)$-implication.
\end{itemize}
\end{lemma}

\begin{proof}
Since $\structB$ is implicationally hard, there exists 
an instance $\instance$ of $\csp(\structB)$ such that $\graphinstance$
contains a path 
$((v,z), C_1) = ((v_1, z_1), D_1), \ldots, ((v_n, z_n), D_n) = ((v, z), C_1)$
such that for all $j \in [n-1]$ the instance $\instance$ contains a constraint $\constraint$
whose scope contains $v_j , z_j, v_{j+1}, z_{j+1}$ and one of the following holds:
\begin{itemize}
\item either the elements $x_1 = v_j, x_2 = z_j = v_{j+1}, x_3 = z_{j+1}$ are pairwise different, 
$\proj_{x_1, x_2, x_3} \constraint = ((x_1, x_2, x_3), R'_j)$ and the relation $R'_j$ is a 
ternary $(\structB, \instance_{v_j, z_j},  \instance_{v_{j+1}, z_{j+1}}, C_j, C_{j+1})$-implication, or 
 \item the elements $x_1 = v_j, x_2 = z_j, x_3 = v_{j+1}, x_4 = z_{j+1}$ are pairwise different, 
$\proj_{x_1, x_2, x_3, x_4} \constraint = ((x_1, x_2, x_3, x_4), R'_j)$ and $R'_j$ is a 
quaternary $(\structB, \instance_{v_j, z_j},  \instance_{v_{j+1}, z_{j+1}}, C_j, C_{j+1})$-implication.
\end{itemize}

It follows that $R' := ((R'_1 \circ R'_2) \circ \ldots \circ R'_k$) is well defined and  by Lemma~\ref{lem:ImplicationsComposition}, the relation $R'$ is a $(\structB, \instance_{v, z},  \instance_{v, z},C_1, C_1)$-implication. Assume first that $R'$ is ternary. If the relation $R'$ is a ternary 
$(\structB, \instance_{v, z}, C_1, \leftarrow,\rightarrow)$-implication, then 
it is also a ternary $(\structB, (\instance_{v, z})^{-1}, (C_1)^{-1}, $
$\rightarrow,\leftarrow)$-implication and then the lemma follows. Otherwise,
% if
%the ternary relation $R'$ is not a $(\leftarrow,\rightarrow)$-implication, then
$R'$ is a $(\rightarrow, \rightarrow)$-implication or a $(\leftarrow, \leftarrow)$ -implication, and then, by Lemma~\ref{lem:CycleImplication}, the relation $R' \circ R'$ is a quaternary  $(\structB, \instance_{v, z}, C_1)$-implication. From now on we may therefore assume that $R'$ is quaternary.

If the relation $R'$ is a $(\rightarrow, \leftarrow)$-implication, then we are done.
Otherwise $R'$ is an $(L,P)$-quaternary implication with either $L \neq \rightarrow$ or $P \neq \leftarrow$. But then 
\begin{itemize}
\item $R'(x_2, x_1, x_3, x_4)$ is a $(\rightarrow, \leftarrow)$-implication in the case where $L = \leftarrow$ and $R = \leftarrow$,
\item $R'(x_1, x_2, x_4, x_3)$ is a $(\rightarrow, \leftarrow)$-implication in the case where $L = \rightarrow$ and $R = \rightarrow$, and
\item $R'(x_2, x_1, x_4, x_3)$ is a $(\rightarrow, \leftarrow)$-implication in the case where $L = \leftarrow$ and $R = \rightarrow$.
\end{itemize}
It completes the proof of the lemma.
\end{proof}

We are now in the position to show that first-order expansions of liberal finitely bounded homogenous binary cores $\structA$ with bounded strict width are implicationally simple and hence their relational width is $(2, \maxbound_{\structA})$.

\begin{lemma}
\label{lem:ImpHardNoQNUF}
Let  $\structA$ be a liberal finitely bounded homogeneous binary core and $\structB$ a first-order expansion of $\structA$ which is implicationally hard. Then $\structB$ has no bounded strict width.
\end{lemma}

\begin{proof}
By Lemma~\ref{lem:CycleImplication}, we may assume that $\structB$ pp-defines an $(\structB, C', C'_1, \rightarrow, \leftarrow)$-implication $R'$. By Lemma~\ref{lem:CompleteImplication}
that $R'$ is a complete implication.  By Lemma~\ref{lem:TwoComponents}, we have that $(\vertices_L(C'_1)$
$ \cup \vertices_R(C'_1))$ is a set of strongly connected components which is a sink in $\bipartite_{R', R'}$. Let $C''_1 \subseteq A^2$ be such that $(\vertices_L(C''_1) \cup \vertices_R(C''_1))$ form a strongly connected component included in $(\vertices_L(C'_1) \cup \vertices_R(C'_1))$ which is a sink in
$\bipartite_{R', R'}$
and $D''_1 \subseteq A^2$ such that $(\vertices_L(D''_1) \cup \vertices_R(D''_1))$ is a strongly connected component which is a source in $\bipartite_{R', R'}$.
Since the component $(\vertices_L(C''_1) \cup \vertices_R(C''_1))$ is a sink, we have that $R'$ is a $(\structB, C', C''_1, \rightarrow, \leftarrow)$-implication,
Now, if either $C'_1$ or $D'_1$ is neither anti-reflexive nor it is $=$, then we replace 
$R'$ with $R'':= R'(x_1, x_2, x_3) \wedge x_1 \neq x_2 \wedge x_2 \neq x_3$ if $R'$ is ternary and with  $R'' := R'(x_1, x_2, x_3, x_4) \wedge x_1 \neq x_2 \wedge x_3 \neq x_4$ if $R'$ is quaternary. 
Set now $C = C' \setminus \{  = \}$, $C_1 = C''_1 \cap C$ and $D_1 = D''_1 \cap C$ and 
observe that $\bipartite_{R'', R''}$ is a digraph induced in $\bipartite_{R', R'}$
 by $(\vertices_L(C) \cup \vertices_R(C))$. It follows that $R''$ is a complete $(\structB, C, C_1, \rightarrow, \leftarrow)$-implication
 where the strongly connected component $(\vertices_L(C_1) \cup \vertices_R(C_1))$ is a sink and the strongly connected component $(\vertices_L(D_1) \cup \vertices_R(D_1))$ is a source in $\bipartite_{R'', R''}$.
 
From now on we will assume that the relation $R$ is a complete $(\structB,C, C_1, \rightarrow, \leftarrow)$-implication where $C_1$ is either $=$ or it is anti-reflexive and $(\vertices_L(C_1) \cup \vertices_R(C_1))$ is a strongly connected component which is a sink in $\bipartite_{R,R}$ as well as that $D_1 \subseteq C \setminus C_1$ is either $=$ or it is anti-reflexive and $(\vertices_L(D_1) \cup \vertices_R(D_1))$ is a strongly connected component which is a source in $\bipartite_{R,R}$. Clearly, either $C_1$ or $D_1$ is anti-reflexive. We will now prove that $R$ pp-defines a ternary
$(\structB,C, C_1, \rightarrow, \leftarrow)$-implication $R_1$ that contains for all $F_1 \subseteq \{ C_1, D_1 \}$ the relation 
$(F_1(x_1, x_2) \wedge F_1^{-1}(x_2, x_3))$. In the case where $R$ is ternary we set $R_1 := (R \bowtie R) \bowtie (R \bowtie R)$ which is equivalent to setting
$R_1 := (R \circ R) \circ (R \circ R)$. Hence, by Lemma~\ref{lem:ImplicationsComposition}, the relation $R$ is a ternary $(\structB,C, C_1, \rightarrow, \leftarrow)$-implication. Further, if $F_1$ is $=$, then clearly 
$R_1$ has a ternary $==$-tuple. If $F_1$ is anti-reflexive, then consider any orbitals $O_1, O_3$ in $F_1$. Since $R$ is a complete implication, by Corollary~\ref{cor:InjTernaryTuple}, it contains  an injective $O_1 O_2^{-1}$-tuple as well as an injective $O_2 O_3^{-1}$-tuple for some orbital $O_2 \subseteq F_1$, and hence by Observations~\ref{obs:AllInjective} and~\ref{obs:ContainsFullBowtie}, we have that $R_1$ contains   $(O_1(x_1, x_2) \wedge O_3^{-1}(x_2, x_3))$, and in consequence that $R_1$ contains $(F_1(x_1, x_2) \wedge F_1^{-1}(x_2, x_3))$.
Clearly, also the relation
$( (R_2(x_1, x_2, x_3) \equiv R_1(x_1, x_2, x_3) \cap R_1(x_3, x_2, x_1)))$
contains $(F_1(x_1, x_2) \wedge F_1^{-1}(x_2, x_3))$ for $F_1 \subseteq \{ C_1, D_1 \}$.
Further, observe that $\bipartite_{R_2, R_2}$
has only symmetric edges. Hence, since $(\vertices_L(D_1) \cup \vertices_R(D_1))$ is a 
sink in $\bipartite_{R, R}$ and also in  $\bipartite_{R_2, R_2}$, we have that
$D_1(x_1, x_2) \equiv R_2(x_1, x_2, x_3) \wedge O^{-1}(x_2, x_3)$ for some orbital contained in $D_1$, and hence $D_1$ is pp-definable in $\structB$.
It follows that $R_2$ is either 
 a critical ternary relation over $(\structB, C_1, D_1, C_1^{-1},  D_1^{-1})$ if $C_1$ is anti-reflexive
or     a critical ternary relation over $(\structB, D_1, C_1, D_1^{-1}, C_1^{-1})$ if $D_1$ is anti-reflexive. By appeal to Proposition~\ref{prop:CriticalTernaryNoQNUF}, the structure $\structB$ does not have bounded strict width.

We now turn to the case where  $R$ is a complete quaternary $(\structB,C, C_1, \rightarrow, \leftarrow)$-implication. This time we look at 
$R_1 := ((R \bowtie R) \bowtie_3 (R \bowtie R))$ and this time we will first show that for all $F_1 \in \{ C_1, D_1 \}$,  $R_1$ contains  the relation $R_F := F_1(x_1, x_2) \wedge F_1^{-1}(x_3, x_4)$. If $F_1$ is $=$ and $R$ contains a constant tuple, then $R_1$ also contains a constant tuple and hence $R_F$. If $F_1$ is $=$ and $R_1$ contains a non-constant $==$-tuple, then by Observations~\ref{obs:AllInjective} and~\ref{obs:ContainsFullBowtieThree}, $R_1$  contains an $==$-tuple and hence $R_F$. For the case where $F_1$ is anti-reflexive we 
again consider any orbitals $O_1, O_3^{-1}$ in $F$. Since $R$ is a complete implication, it contains  an injective $O_1 O_2^{-1}$-tuple as well as an injective $O_2 O_3^{-1}$-tuple for some some orbital $O_2 \subseteq F_1$. It follows by Observations~\ref{obs:AllInjective} and~\ref{obs:ContainsFullBowtieThree} that $R_1$ contains $R_F$ also in the case where $F_1$ is anti-reflexive. 
Consider now $(R_2(x_1, x_2, x_3) \equiv R_1(x_1, x_2, x_3) \cap R_1(x_3, x_2, x_1))$ and obeserve that $\bipartite_{R_2, R_2}$ has only symmetric edges.
Since $(\vertices_L(D_1) \cup \vertices_R(D_1))$ is a strongly connected component in $\bipartite_{R_2, R_2}$ we can pp-define $D_1$.
Hence $R_2$ is either 
 a critical ternary relation over $(\structB, C_1, D_1, C_1^{-1},  D_1^{-1})$ if $C_1$ is anti-reflexive
or     a critical ternary relation over $(\structB, D_1, C_1, D_1^{-1},  C_1^{-1})$ if $D_1$ is anti-reflexive. By Proposition~\ref{prop:CriticalTernaryNoQNUF}, we have that $\structB$ does not have bounded strict width.
\end{proof}

\noindent
We are now ready to provide the proof for Theorem~\ref{thm:main}.

\paragraph{Proof of Theorem~\ref{thm:main}}
The result follows by Lemma~\ref{lem:ImpHardNoQNUF} and Theorem~\ref{thm:ImpSimpleMinimal}.

\section{Future Work}

In this paper we characterized the relational width of first order expansions of liberal finitely bounded homogenous binary cores with  bounded strict width. First of all it is natural to ask what happens if the binary core is not liberal.  By Proposition~\ref{prop:qnuExTwoEquivClass} we have that the method in this paper does not work when we forbid structures of size $3$. On the other hand, by the result in~\cite{WronaSTACS2020} it holds that the relational width of first-order expansions of homogenous graphs $\structA$ even with $\maxbound_{\structA} = 3$ is still $(2, \maxbound_{\structA})$. Could it be true for all binary cores $\structA$?

Last but not least, what happens to the relational width when we allow first-order expansions of cores $\structA$ of arbitrary arity $k$? We believe that with an appropriate notion   of liberal $k$-ary cores $\structA$ one can show that their first-order expansions $\structB$ with bounded strict width have relational width $(k, \maxbound_{\structA})$. Could it be true for all reducts of finitely bounded homogeneous structures?

%% Acknowledgments
%\begin{acks}                            %% acks environment is optional
                                        %% contents suppressed with 'anonymous'
  %% Commands \grantsponsor{<sponsorID>}{<name>}{<url>} and
  %% \grantnum[<url>]{<sponsorID>}{<number>} should be used to
  %% acknowledge financial support and will be used by metadata
  %% extraction tools.
%  This material is based upon work supported by the
%  \grantsponsor{GS100000001}{National Science
%    Foundation}{http://dx.doi.org/10.13039/100000001} under Grant
%  No.~\grantnum{GS100000001}{nnnnnnn} and Grant
%  No.~\grantnum{GS100000001}{mmmmmmm}.  Any opinions, findings, and
%  conclusions or recommendations expressed in this material are those
%  of the author and do not necessarily reflect the views of the
%  National Science Foundation.
%\end{acks}

%% Bibliography
\bibliographystyle{plain}
\bibliography{mybib_LICS.bib}

\end{document}